\documentclass[12pt]{amsart}
\usepackage{amsthm}
\usepackage{geometry}
\usepackage{fullpage}
\usepackage[utf8]{inputenc}
\usepackage{comment}
\usepackage{amsmath, amssymb}
\usepackage{mathrsfs, enumerate, csquotes, color,float}
\newcommand{\abs}[1]{\left\lvert#1\right\rvert}
\usepackage[hidelinks]{hyperref}
\usepackage{cleveref}
\usepackage{tikz,pgfplots}
\pgfplotsset{%
	every x tick/.style={black, thin},
	every y tick/.style={black, thick},
	every tick label/.append style = {font=\footnotesize},
	every axis label/.append style = {font=\footnotesize},
	compat=1.12
}
\newcommand{\rev}[1]{{\color{black}#1}}

\newcommand{\R}{\mathbb{R}}
\DeclareMathOperator{\supp}{supp}
\DeclareMathOperator{\dist}{dist}
\renewcommand{\Re}{\mathrm{Re}\,}

\renewcommand{\leq}{\leqslant}	
\renewcommand{\geq}{\geqslant}

\DeclareMathOperator{\Id}{Id}
\newcommand{\N}{\mathbb{N}}

\newcommand{\C}{{\mathbb C}}

\newtheorem{theorem}{Theorem}[section]

\newtheorem{lemma}[theorem]{Lemma}
\newtheorem{corollary}[theorem]{Corollary}
\newtheorem{proposition}[theorem]{Proposition}
\theoremstyle{definition}
\newtheorem{definition}[theorem]{Definition}
\newtheorem{remark}[theorem]{Remark}

\makeatletter

\@addtoreset{equation}{section}
\makeatother

\let\oldtocsection=\tocsection

\let\oldtocsubsection=\tocsubsection

\let\oldtocsubsubsection=\tocsubsubsection

\renewcommand{\tocsection}[2]{\hspace{0em}\oldtocsection{#1}{#2}}
\renewcommand{\tocsubsection}[2]{\hspace{1em}\oldtocsubsection{#1}{#2}}
\renewcommand{\tocsubsubsection}[2]{\hspace{2em}\oldtocsubsubsection{#1}{#2}}

\title{Semiclassical tunneling for some 1D Schrödinger operators with complex-valued potentials}

\author[M. Averseng]{M. Averseng}
\address[M. Averseng]{Univ Angers, CNRS, LAREMA, F-49000 Angers, France}
\email{martin.averseng@univ-angers.fr}

\author[N. Frantz]{N. Frantz}
\address[N. Frantz]{Univ Angers, CNRS, LAREMA, F-49000 Angers, France}
\email{nicolas.frantz@univ-angers.fr}

\author[F. H\'erau]{F. H\'erau}
\address[F. H\'erau]{LMJL - UMR6629, Nantes Université, CNRS, 2 rue de la Houssini\`ere, BP 92208, F-44322 Nantes cedex 3, France}
\email{herau@univ-nantes.fr}

\author[N. Raymond]{N. Raymond}
\address[N. Raymond]{Univ Angers, CNRS, LAREMA, Institut Universitaire de France, F-49000 Angers, France}
\email{nicolas.raymond@univ-angers.fr}

\begin{document}
	
%
%

\begin{abstract}
We consider the non-selfadjoint, semiclassical Schrödinger operator $\mathscr{L}(h) := -h^2\partial_x^2+e^{i\alpha}V$, where $\alpha \in (-\pi,\pi)$ and $V: \R\to \R_+$ is even and vanishes at exactly two (symmetric) non-degenerate minima. We establish a semiclassical tunneling result: the spectrum of $\mathscr{L}(h)$ near the origin is given by a sequence of algebraically simple eigenvalues which come in exponentially close pairs (within a $\mathscr{O}(e^{-S/h})$ distance where $S > 0$ is explicit), each pair being separated from the others by a distance $\mathscr{O}(h)$. A one-term estimate of the gap between the two smallest eigenvalues in magnitude is derived; it reveals that, when $\alpha \neq 0$, they quickly rotate around each other as $h$ goes to $0$. 
\end{abstract}

	\maketitle
%
%
%


\section{Introduction}

\subsection{Main result}

In an influential series of works \cite{HS84,HS85}, Helffer and Sjöstrand (see \S\ref{sec:literature} for a more throrough review of the literature) established results about the spectrum of semiclassical Schrödinger operators, a particular case of which can be summarized as follows.

Let $V : \mathbb{R} \to \mathbb{R}$ be a smooth potential which admits a minimum not attained at infinity and possesses the even symmetry $V(x) = V(-x)$. Moreover, suppose that $V$ admits exactly two global, non-degenerate minima at $x_\ell$ and $x_r = -x_\ell$. Then, for $h > 0$ small enough, the low-lying spectrum of the semiclassical Schrödinger operator 
$$\mathscr{L}(h) := -h^2 \frac{d^2}{dx^2} + V(x): H^2(\R^d) \to L^2(\R^d)$$
consists of pairs of exponentially close but simple eigenvalues, each pair separated from the others by a distance of order $h$; furthermore, the gap between the first two eigenvalues $\mu_1(h)$ and $\mu_2(h)$ satisfies the estimate 
\begin{equation}\label{eq.0}
	\mu_2(h) - \mu_1(h) =  \big(\mathsf{A} + o_{h \to 0}(1) \big)h^{1/2} e^{-{S}/h}\,,
	\end{equation}
where $\mathsf{A}>0$ is an explicit constant and $S>0$ is the so-called ``Agmon distance"\footnote{after S. Agmon, who introduced this distance to study localization properties of the solutions of Schrödinger's equation; see, e.g. \cite{agmon1982lectures}} between the minima, defined by
$$S = d_{V}(x_\ell,x_r) = \int_{x_\ell}^{x_r} \sqrt{V(s)}\,ds.$$
The gap between the first two eigenvalues is connected to {\em quantum tunneling} (see e.g. \cite{razavy2013quantum}) in that after a time $T(h) \sim \frac{\pi h}{\mu_2(h) - \mu_1(h)}$, low-energy quantum states that were initially localized with strong probability in the left well (near $x_\ell$) will typically be localized with strong probability in the right well, thus ``passing through'' the potential barrier between the wells.  

In the present work, our aim is to extend the above result to a situation featuring a {\em complex-valued potential}. Namely, given $\alpha \in (-\pi,\pi)$, we wish to describe the low-energy spectrum of
\begin{equation}
\label{eq:op}
\mathscr{L}(h) = -h^2 \frac{d^2}{dx^2} + e^{i\alpha} V(x) : H^2(\R) \to L^2(\R)
\end{equation}
where $V: \R \to \R_+$ is as above (in fact, our main result holds for more general complex potentials, see Remark \ref{rem:comments_thm_gap} below). 

Naturally, the main challenge is that, for $\alpha \neq 0$, $\mathscr{L}(h)$ is a non-selfadjoint operator, and its spectrum could in principle differ radically from the selfadjoint case, even when $|\alpha|$ is small. Indeed, the spectral properties of non-selfadjoint operators are known to be highly sensitive to small perturbations, see for instance \cite{davies1982pseudo,davies2007linear}. Moreover, for $\alpha \neq 0$, it is not easy to guess what becomes of the gap estimate \eqref{eq.0}. Readers familiar with the analysis of the selfadjoint setting might expect that the groundstates $\psi_{\ell/r}$ attached to the left/right well may become oscillatory for $\alpha \neq 0$, and that some destructive interference could cause the ``interaction term'' $\langle \psi_\ell,\psi_r\rangle$ (which is directly related to the gap \eqref{eq.0} in the selfadjoint setting) to be significantly smaller than in the selfadjoint case. Does this lead to a qualitatively different phenomenon? Is the eigenvalue gap still related to some (generalized) Agmon distance between the wells? The answer to these questions, and the main result of this paper, is the following.
\begin{theorem}
\label{thm:gap}
Let $V : \R \to \R_+$ be smooth, even and such that
$$\liminf_{x\to\pm\infty} V(x) > 0.$$
Suppose that $V$ vanishes exactly at two points $x_\ell$ and $x_r= -x_\ell$, and that these minima are non-degenerate, i.e., $V''(x_\ell) = V''(x_r) > 0$. Given $\alpha \in (\pi,\pi)$, let $\mathscr{L}(h)$ be the unbounded operator on $L^2(\R)$ with domain $H^2(\R)$, defined by
$$\mathscr{L}(h) := -\frac{d^2}{dx^2} + e^{i\alpha}V.$$
There exists $h_0 > 0$ such that for all $h \in (0,h_0)$, $\mathscr{L}(h)$ admits two distinct eigenvalues $\mu_1(h)$ and $\mu_2(h)$ which are the smallest in modulus. They are algebraically simple and their difference satisfies the estimate
\begin{equation}
\label{eq:gap}
\mu_2(h) - \mu_1(h) =  \big(\mathsf{A} + o_{h \to 0}(1) \big)h^{1/2} e^{-{S(\alpha)}/h},
\end{equation}
where 
\begin{equation}
\label{eq:defSalpha}
S(\alpha) := e^{i \alpha/2} \int_{x_\ell}^{x_r} \sqrt{V(x)}\,dx, \quad \textup{and}
\end{equation}
\begin{equation}
\label{eq:defA}
\mathsf{A} := 4e^{i3\alpha/4}\sqrt{\frac{V(0)}{\pi}} \left(\frac{V''(x_\ell)}{2}\right)^{1/4}\exp\left(-2\int_{x_\ell}^0 \frac{(\sqrt{V})'(s) -  \sqrt{V''(x_\ell)/2}}{\sqrt{V(s)}}\,ds\right)\,.
\end{equation}
Moreover, there exists $C >0$ such that 
\begin{equation}
\label{eq:separation}
|\mu(h)-\mu_1(h)| \geq Ch
\end{equation}
for all $h \in (0,h_0)$ and any $\mu(h) \in \textup{sp}(\mathscr{L}(h)) \setminus \{\mu_1(h),\mu_2(h)\}$.
\end{theorem}

\begin{remark}[Comments on Theorem \ref{thm:gap}]
\label{rem:comments_thm_gap}\
\begin{enumerate}[\rm (i)]
\item When $\alpha=0$, \eqref{eq:gap} agrees with the well-known tunneling formula of \cite{HS84}.
\item Inspecting the proof of Theorem \ref{thm:gap}, one can check that it still holds in the more general case where the potential is of the form $e^{i\alpha(x)}V$, with $x\mapsto\alpha(x)$ a $C^\infty$ function such that $\alpha(\R) \subset [-\pi+\varepsilon,\pi-\varepsilon]$ for some $\varepsilon>0$. In this case, $S(\alpha)$ should be replaced by $\int_{x_\ell}^{x_r} e^{i\alpha(s)/2}\sqrt{V(s)}ds$, and $e^{3i\alpha/4}$ by $e^{i(\alpha(x_\ell)+2\alpha(0))/4}$.
\item When $\alpha\neq 0$, the smallest two eigenvalues are rotating rather quickly around each other since $\arg(\mu_2(h)-\mu_1(h))\sim \rev{-}S\sin(\alpha/2)/h$.
\item We have $|\mu_2(h)-\mu_1(h)|
=\rev{|\mathsf{A}|}\sqrt{h}e^{-S\cos(\alpha/2)/h}(1+o(1))$	with $S= \int_{x_\ell}^{x_r} \sqrt{V(s)}\,ds$. Thus, the magnitude of the eigenvalue gap {\em increases} when $\alpha$ increases.
\item In particular, the size of the eigenvalue gap does {\em not} collapse due to any kind of destructive interactions. \rev{This is in contrast with other contexts where such destructive interferences are known to occur (see, e.g., the recent article \cite{fefferman2025magnetic} involving a magnetic field).} \rev{Here}, we show that the tunneling amplitude is in fact related to the quantity $\langle\psi_\ell,\overline{\psi}_r\rangle$  where $\psi_\ell$ and $\psi_r$ are as above (and {\em not} $\langle\psi_\ell,{\psi}_r\rangle$, which does become significantly smaller when $\alpha \neq 0$). This essentially comes from the fact that $J\mathscr{L}(h)=\mathscr{L}(h)^*J$ where $J$ is the complex conjugation. Let us point out that such a property, where $J$ is a more general, abstract conjugation operator, has been used for the analysis of non selfadjoint operators, see \cite[Hypothesis 4]{FAFR24} and \cite[Hypothesis 4]{Fr25}. 
\item In this paper, the complex potential $W = e^{i\alpha}V$ satisfies the symmetry $W(x) = W(-x)$. One could also consider a case of the $\mathcal{PT}$-symmetry, i.e., where $W(x) = \overline{W(-x)}$ (taking $\alpha = \alpha(x)$ as an odd function). One would obtain a similar gap formula with some other prefactor $\mathsf{B}$ instead of $\mathsf{A}$ and $\textup{Re}(S(\alpha))$ instead of $S(\alpha)$.   
\end{enumerate}
\end{remark}

\rev{\begin{remark}[Dynamical interpretation]
When $\alpha \neq 0$, though the operator $\mathscr{L}(h)$ is non-selfadjoint (and non-unitary), there is still a dynamical interpretation of the eigenvalue gap. Indeed, let us consider 
$$ih\partial_t \psi = \mathscr{L}(h)\psi$$
with an initial condition 
$$\psi(0,\cdot) = \psi_1 + \psi_2$$
where $\psi_j$ is an $L^2$ normalized eigenfunction of $\mathscr{L}(h)$ associated to $\mu_j(h)$, $j = 1,2$. Since $\mathscr{L}(h)$ commutes with the symmetry $\psi(x) \to \psi(-x)$, and since $\mu_j(h)$ are algebraically simple, one has
$$\psi_j(-x) = e^{i \theta_j}\psi_j(x)$$ 
for some $\theta_j \in \R$, $j = 1,2$. In particular, $\psi_1$ has the same $L^2$ mass in each well, and the same is true for $\psi_2$. Moreover, as in the self-adjoint case, the results of this paper show that $\psi_1$ and $\psi_2$ can be chosen so that the linear combination $\psi_1 \pm \psi_2$ is exponentially localized in a $O(\sqrt{h})$ neighborhood of the left/right well. Depending on the sign of $b = \textup{Im}(\mu_2(h) - \mu_1(h))$, the solution for $t > 0$ (and $h > 0$ fixed small enough) satisfies
\begin{align*}
\psi(x,t) 
&= e^{\frac{t \mu_1(h)}{ih}} \psi_1+  e^{\frac{t \mu_2(h)}{ih}} \psi_2\\
& \sim \begin{cases}
 e^{\frac{t \mu_1(h)}{ih}} \psi_1 & \textup{if } b < 0,\\ 
e^{\frac{t\mu_2(h)}{ih}}\psi_2 & \textup{if } b > 0,\\
 e^{\frac{t \textup{Im}(\mu_1(h))}{h}} \left(\psi_1 + e^{\frac{i t\textup{Re}(\mu_2(h) - \mu_1(h))}{h}}\psi_2\right) & \textup{if } b = 0,
\end{cases}
\end{align*}
as $t \to \infty$ (with equality when $b = 0$). Thus, for $b \neq 0$, the system, initially localized in the left well, approaches a state with equal $L^2$ mass in each well, while for $b = 0$, one recovers an oscillatory behavior analogous to the selfadjoint case. 
\end{remark}}

\subsection{Context and motivation}
\label{sec:literature}
The analysis of quantum tunneling for Schrödinger operators with real-valued, multiple-well potentials was started in dimension one by Harell in 1980 \cite{Harrell}, followed by 
results in arbitrary dimension by Simon \cite{Simon84,Simon84b} and Helffer and Sjöstrand \cite{HS84,HS85}; for an introduction to these results, we refer to the monographs \cite{DiSj99,Hel88}, the synthetic presentation of \cite{Robert} (in french) or \cite{BHR17} for a pedagogical treatment of a connected problem in dimension one. Since these pioneering papers, a large body of works has been devoted to establishing tunneling formulas for other kinds of operators, including for instance magnetic fields, see e.g. \cite{BHR17,FBMR25} or electromagnetic fields, see e.g. \cite{fefferman2022lower,helffer2022quantum,morin2024tunneling}. In all of these cases, the operator under consideration is selfadjoint.
 
Complex potential barriers are also relevant in the physics literature as a model for absorption (see e.g. \cite{muga2004complex} for a review) and tunneling times in such potentials have also been investigated by physicists (see e.g. \cite{raciti1994complex,kovcinac2008tunneling} and the references therein). However, there seem to be comparatively fewer mathematical works establishing analogous tunneling results for non-selfadjoint operators: the only examples that we aware of are those connected to the Krammers-Focker-Planck operator, see e.g. \cite{herau2008tunnel,herau2011tunnel}. Let us also mention that the problem of bounding/counting the eigenvalues of (semiclassical) Schrödinger operators with complex-valued potentials has attracted attention in the past decade, see, e.g., \cite{frank2011eigenvalue1,frank2017eigenvalue2,frank2018eigenvalue3,cuenin2024open}.

Schrödinger operators with complex potentials also appear in the recent mathematical literature for problems set on domains $\Omega \subset \R^d$ carrying Dirichlet conditions, see e.g. \cite{A08, Almog-Helffer-Pan_2012, Almog-Helffer-Pan_2013, AH16, Almog-Helffer_2020, HKR24}. Most of these works are motivated by the theory of superconductivity and aim at estimating the decay of the semigroup $(e^{-t \mathscr{L}(h)})_{t\geq 0}$ by giving lower bounds on the real part of the spectrum. In \cite{HKR24} (extending the case $\alpha=0$ analyzed in \cite{CKPRS22}), families of eigenvalues are accurately described in some regions of the complex plane. There, it is established that eigenfunctions under consideration are exponentially localized near the boundary, allowing to reduce the spectral analysis to two-dimensional operators quite similar to \eqref{eq:op} on the half-plane. These can be analyzed by separation of variables. In the selfadjoint case \cite{CKPRS22}, it is even possible to prove an \emph{optimal} localization of the eigenfunctions near specific points of the boundary. However, in the non-selfadjoint situation \cite{HKR24}, this optimality is lost due to the complex scaling argument used there, which does not allow to control the tangential localization of the eigenfunctions. At the core of the problem are operators on the boundary in the form \eqref{eq:op} (when $d=1$) and the optimal localization behavior of their eigenfunctions.

Let us also emphasize that the case of dimension one is fundamental, since several spectral problems in higher dimensions can be reduced to lower dimensions by means of Feshbach/Grushin methods (see \cite{Keraval, BFRV25}). For example, the generalizations of the result outlined in the introduction for one-dimensional pseudodifferential operators recently obtained in \cite{DR25}, enabled, via a microlocal dimensional reduction, to establish the first known tunneling formula for a double well magnetic field in dimension two, see \cite{FBMR25}. 

Let us also mention that the analysis of exponentially small effects for 1D non selfadjoint operators, in the context of the Bohr-Sommerfeld rule, has given rise to the recent article \cite{hitrik2024overdamped} (see especially Section 5 of that reference) which is motivated by the equations of the relativity; on a closely connected topic, see also \cite{Duraffour}.

\subsection{Outline of the proof of Theorem \ref{thm:gap}}

As in the selfadjoint case, the proof of Theorem \ref{thm:gap} relies on a decoupling of the wells near $x = x_\ell$ and $x = x_r$, and a very precise computation of the interaction between the eigenfunctions related to each well (see Section \ref{sec:double_well1D}). Such a precise computation is made possible thanks to WKB-type approximations\footnote{after Wentzel, Kramers and Brillouin \cite{wentzel1926verallgemeinerung,kramers1926wellenmechanik,brillouin1926remarques}, see also the early review \cite{dunham1932wentzel}.} of the eigenfunctions near each well, see Section \ref{sec:proofWKB}. The proof of these WKB approximations takes two steps: first, one constructs exponentially good quasi-solutions of the eigenvalue equation by solving a series of ordinary differential equations, and second, one uses bounds on the resolvent for the simple-well operators to deduce that these quasi-solutions actually close to true eigenfunctions. In the selfadjoint case, the bounds on the resolvent needed in the second step are readily available thanks to the spectral theorem, but in the present, non-selfadjoint setting, a replacement is needed. This is the object of Section \ref{sec:3}. The main idea is that, due to the exponential localization results shown in Section \ref{sec:elliptic1D}, the resolvent of the simple-well operator is close (near its poles in $D(0,Rh)$) to that of a rescaled complex harmonic oscillator, which, despite still being non-selfadjoint, is sufficiently well understood. In more details: 
\begin{itemize}
\item[---] In Section \ref{sec:elliptic1D}, we define the {\em left-well operator} $\mathscr{L}_\ell(h)$, obtained by ``sealing'' the potential well at $x_r$ (that is, replacing $V$ by $V_\ell := V + \Sigma_\ell$, $\Sigma_\ell \geq 0$, with $\Sigma_\ell > 0$ near $x_r$ and $\Sigma \equiv 0$ near $x_\ell$). We analyse its low-energy eigenfunctions, i.e., those associated to eigenvalues in a small disk $D(0,Rh)$ in the complex plane; we show that they are exponentially localized near $x_\ell$ (see Corollary \ref{cor:small_far}), and as a consequence, that they provide (i) exponentially accurate (in the limit $h \to 0$) quasi-modes for the double-well operator $\mathscr{L}(h)$ as well as (ii) $\mathscr{O}(h^{3/2})$ quasi-modes for its quadratic approximation near $x_\ell$, which, up to a rescaling, is a complex harmonic oscillator (see Corollary~\ref{cor:quasimodes}).

To prove the aforementioned exponential localization, the key tool is an elliptic estimate for the conjugated operator $P^{\Phi}(h)=e^{\Phi/h}P(h)e^{-\Phi/h}$ where $P(h)$ is an operator of the general form $P(h)=h^2D_x^2+e^{i\alpha}U$ and $\Phi$ is a convenient real-valued function (see Proposition \ref{prop:estimee_elliptique_conj}). Its proof relies on the fact that, for all $\xi\in\R$,
\begin{equation}\label{eq.key}
\textup{Re} \big(e^{-i\alpha/2}\left[(\xi+i\Phi'(x))^2+e^{i\alpha}U(x)\right]\big)\geq\frac{1}{\cos(\alpha/2)}(\cos^2(\alpha/2)U(x)-\Phi'^2(x))\,.
\end{equation}
This decay estimate is applied for $U := V_\ell$ (the potential sealed on the right). Roughly speaking, the one-well eigenfunctions decay like $e^{-\Phi_\ell(x)/h}$, where $\Phi_\ell(x)=\left|\int_{x_{\rev{\ell}}}^{x}\cos(\alpha/2)\sqrt{V_\ell(s)}ds\right|$ cancels out the right-hand-side of \eqref{eq.key}). The elliptic estimate is used again in Section \ref{sec:proofWKB} with a refined choice for the weight $\Phi$. 
\item[---] Section \ref{sec:3} is devoted to the analysis of the resolvent of the left-well operator $(\mathscr{L}_\ell(h) - z)^{-1}$ for $z \in D(0,Rh)$, and more specifically, the asymptotic location of its poles -- i.e., the eigenvalues of $\mathscr{L}_\ell(h)$ -- $\mu_{\ell,1}(h), \ldots, \mu_{\ell,n}(h)$, and its behaviour near them. The results are formulated in terms of the properties of the Riesz projectors 
$$\Pi_{\ell,n}(h) := \int_{\gamma_n(h)} (z - \mathscr{L}_\ell(h))^{-1}dz,$$
where $\gamma_n(h)$ is a closed contour in $\mathbb{C} \setminus \textup{sp}(\mathscr{L}_\ell(h))$ circling around a pole $\mu_{\ell,n}(h)$, see Proposition~\ref{prop:fourre_tout}. Roughly speaking, these results allow to get around the impossibility of applying the spectral theorem to the non-selfadjoint operator $\mathscr{L}_\ell(h)$.

The key ingredients in the proof of Proposition \ref{prop:fourre_tout} are some properties of the complex harmonic oscillator $\mathscr{H}(h)$ and its resolvent, recapped/established in \S \ref{sec:davies}. In relating the properties of $(\mathscr{H}(h) - z)^{-1}$ and $(\mathscr{L}_\ell(h)-z)^{-1}$, the fact established in Section \ref{sec:elliptic1D} that low-energy eigenfunctions of $\mathscr{L}_\ell(h)$ are $\mathscr{O}(h^{3/2})$ quasimodes for a rescaled complex harmonic oscillator, naturally plays a key role.

\item[---] In Section \ref{sec:proofWKB} we give WKB approximations of the low-energy eigenfunctions and associated eigenvalues of the left-well operator (or, by symmetry, the analogous right-well operator). They take the form $\psi^{\mathrm{wkb}}_n(x;h)=e^{-\varphi_\ell(x)/h}a_n(x;h)$ ($n\geq 1$), with $\varphi_\ell(x)=e^{i\frac{\alpha}{2}}\left|\int_{x_\ell}^x\sqrt{V(s)}ds\right|$; in particular, they decay like $e^{-\Re\varphi_\ell(x)/h}=e^{-\Phi_\ell(x)/h}$, that is, precisely like the eigenfunctions of $\mathscr{L}_\ell(h)$. The corresponding quasi-eigenvalues are $(2n-1)e^{i\alpha/2}\sqrt{\frac{V''(x_\ell)}{2}}h$ modulo a remainder of order $\mathscr{O}(h^2)$. The main result is Proposition \ref{prop:WKB_estim}: it shows that these WKB constructions are optimal in the sense that they describe exactly the spectrum of $\mathscr{L}_\ell(h)$ in $D(0,Rh)$ and that the WKB Ansätze are exponentially good approximations of the eigenfunctions, see \eqref{eq.approxWKB}. Besides the resolvent bounds of Section \ref{sec:3}, the main idea is to use again the elliptic estimate of Proposition \ref{prop:estimee_elliptique_conj}, this time with a family of well-crafted ``subsolutions'' (in the sense of Definition \ref{def:M_solution}), see Lemma \ref{lem:optimal_weight}. These subsolutions are very close to $\Re \varphi_\ell$, and their expression is inspired by \cite{DR25}. 

\item[---] Section \ref{sec:double_well1D} is devoted to the proof of Theorem \ref{thm:gap}. The first step is to prove that the spectrum of $\mathscr{L}(h)$ in $D(0,Rh)$ is made of duets of exponentially close eigenvalues (modulo $\mathscr{O}(e^{-(\Re S(\alpha)-\delta)/h})$), see Proposition \ref{prop:precise_spectrum_L}. The key element to establish this is the appproximation of the Riesz projector of $\mathscr{L}(h)$ by the sum of the Riesz projectors associated with $\mathscr{L}_\ell(h)$ and $\mathscr{L}_r(h)$ (the left- and right-well operators) up to a remainder of order $\mathscr{O}(e^{-(\Re S(\alpha)-\delta)/h})$, see Proposition \ref{prop:PiPilPir}. The spectral gap is analyzed in Section \ref{sec:gap}, see Proposition \ref{prop:gap1} and Lemma \ref{lem:wronskien}, the final estimate of Lemma \ref{lem:valeurW(0)} crucially using the exponential WKB approximation. Note that the estimate of the gap in Proposition \ref{prop:gap1} deviates from the expression available in the selfadjoint case. See especially the presence of $\overline{\psi}_r$ (and not $\psi_r$) in the formula, which comes from Corollary~\ref{cor:approxPistar}.
\end{itemize}

\subsection{Notation}

	In what follows, we write
	$D_x := -i \frac{d}{dx}$. We denote by $C^\infty(U)$ the space of infinitely differentiable functions on $U$ and $C^\infty_c(U)$ the subspace of $C^\infty(U)$ consisting of functions with compact support in $U$. The Schwartz class is denoted by $\mathscr{S}(\R)$, and $H^m(\R)$ and $H^m_h(\mathbb{R})$ will denote the Sobolev space and the semi-classical version of it with the norms
	$$\|u\|^2_{H^m_h} := \sum_{l = 0}^m \|(hD_x)^l u\|^2\,, \quad \|u\|^2_{H^m} :=  \sum_{l = 0}^m \|D_x^l u\|^2\,,$$
	where $\|\cdot\|$ is the $L^2$ norm over $\mathbb{R}$.
	The notation $\|A\|$ will also be used for the norm of the operator $A:L^2(\mathbb{R}) \to L^2(\mathbb{R})$. As usual, we denote by $[A,B] := AB - BA$ the commutator of $A$ and $B$, whenever it makes sense. Given $U \subset \R$ open, we write 
	$$\|u\|_{W^{n,\infty}(U)} := \max_{0 \leq j \leq n} \sup_{U} |D_x^j u|.$$
	For $\chi,\underline{\chi}$ and $\widetilde{\chi} \in C^\infty(\mathbb{R})$, we adopt the notation
	$$\chi \perp \widetilde{\chi} \iff \dist(\supp \chi,\supp \widetilde{\chi}) > 0\,, \quad \chi \prec \underline{\chi} \iff \chi \perp (1 - \underline{\chi})$$
	We will consistently denote by $\chi_\ell \prec \underline{\chi_\ell}$, etc. functions that are equal to $1$ in a neighborhood of $x_\ell$ (functions localized ``near the left well'') and by $\widetilde{\chi}_\ell \prec \underline{\widetilde{\chi}_\ell}$, etc. functions that vanish in a neighborhood of $x_\ell$ (functions localized ``away from the left well''). 
	
	We denote by $\textup{sp}(A)$ the spectrum of the (possibly unbounded) operator $A$. For a complex number $z$ and $r > 0$, we denote by 
	$D(z,r) \subset \mathbb{C}$ the open disk of radius $r$ centered at $z$, and by $\mathscr{C}(z,r)$ its boundary. In the proofs, we use the letter $C$ in estimates like $a \leq Cb$ to denote a generic constant whose value may change from one line to another, but does not depend on the universally quantified variables of the statement. Finally, sometimes, the dependence of a symbol in the semiclassical parameter $h$ will be omitted to alleviate the notation.

	\section{An elliptic estimate and its application to exponential decay}
	\label{sec:elliptic1D}

	\subsection{The elliptic estimate}
	
	We first prove an elliptic estimate for operators of the form 
	$$P(h) = (hD_x)^2 + e^{i\alpha} U,$$
	where $U : \mathbb{R} \to \mathbb{R}_+$ is locally integrable and bounded. 
	For this, we introduce
	$$P^\Phi(h) := e^{\Phi/h}P(h)e^{-\Phi/h},$$
	where $\Phi$ is some twice-differentiable function. Solving $(P^\Phi(h) -z)w = 0$ amounts to seeking a solution of $(P(h) - z) u= 0$ under the form $u = e^{-\Phi/h}w$, and elliptic estimates for $(P^\Phi(h)-z)^{-1}$ provide bounds on the norm of $w = e^{\Phi/h}u$ which translate into exponential decay informations for the solution $u$. Writing
	\begin{equation}
	P^\Phi(h)
	= (hD_x + i \Phi')^2 + e^{i\alpha}U = h^2D_x^2 + ih(D_x \Phi' + \Phi'D_x) + e^{i\alpha}U- \Phi'^2;
	\label{eq:expression_Lphi}
	\end{equation}
	and cancelling the leading order $h = 0$ leads to the {\em eikonal equation} (see also \eqref{eq:defPhil} below)
	\begin{equation}
	\label{eq:eikonal}
	\Phi'^2 = e^{i\alpha} U.
	\end{equation}
	A good choice for $\Phi$ is one that almost solves the eikonal equation while still leaving $P^\Phi(h)$ elliptic. This motivates the following definition.
	
	\begin{definition}[$\kappa$-subsolution]
		\label{def:M_solution}
		Let $\kappa > 0$, $\Phi$ a real-valued, non-negative, twice differentiable function, and $\chi$ a smooth, compactly supported function with $0 \leq \chi \leq 1$. We say that $\Phi$ is  \textit{a $\kappa$-subsolution associated to $\chi$ for the potential $U$} if
		$$\Phi'^2 \leq \cos(\alpha/2)^2 U - \kappa (1 - \chi^2).$$
	\end{definition}
	
	\begin{proposition}[Elliptic estimate]
		\label{prop:estimee_elliptique_conj}
		\rev{Let $R > 0$ and $h_0 > 0$.} Then there exists $C(R,h_0) > 0$ such that the inequality
		\begin{equation} 
		\label{eq.elliptic_estimate}
		\|u\|_{H^2_h} \leq C(R,h_0) \left(h^{-1}\|(P^{\Phi}(h)-z)u\|+\|\chi u\|\right)
		\end{equation}
		holds for any $h \in (0,h_0)$, $z \in D(0,Rh)$, $u \in H^2(\R)$, $\chi \in C^\infty(\R,[0,1])$, and any  $(Mh)$-subsolution $\Phi$ associated to $\chi$ for the potential $U$ \rev{with $M > 3R$}.  
	\end{proposition}
	\begin{proof} 
		Since $D_x$ is symmetric in the $L^2$ inner product, one has
		\begin{equation}
			\label{eq.DxPhiPhiDx}
			\langle (D_x \Phi' + \Phi'D_x)u,u\rangle \in \mathbb{R},
		\end{equation}
		for all $u \in H^2(\R)$. Using \eqref{eq:expression_Lphi},
		we deduce that
		$$ \Re \langle(P^\Phi(h) - z) u,u\rangle = \|(hD_x)u\|^2 + \int_\R (\cos(\alpha) U - \Phi'^2)|u|^2.$$
	 	Hence,
		$$\|(hD_x)u\| \leq C \left(\big\|(P^\Phi(h) - z)u\big\| + \|u\|\right).$$
		In turn, the expression~\eqref{eq:expression_Lphi} 
		gives
		$$\|(hD_x)^2u\| \leq C\left( \|(P^\Phi(h) - z)u\| + \|u\|\right),$$
		and it remains to show that 
		\begin{equation}
			\label{eq:sufficientAgmon}
			\|u\| \leq Ch^{-1}\|(P^\Phi(h)-z)u\|+C\|\chi u\|.
		\end{equation}
		
		To this end, we use \eqref{eq.DxPhiPhiDx} again to see that
		\begin{align*}
			\Re\langle e^{-i \frac{\alpha}{2}} P^\Phi(h) u,u\rangle &  = \cos(\alpha/2)\left(\|(h D_x)u \|^2+\int_\R(U - \Phi'^2)|u|^2\right)\\
			& \hspace{4cm} + 2\sin(\alpha/2)\Re\langle (h D_x) u, \Phi'u\rangle.
		\end{align*}
		Next, the estimate
		$$2\Re \big\langle (h D_x) u ,\Phi' u \big\rangle \geq -\frac{\|\Phi' u\|^2}{\varepsilon}-\varepsilon \|(hD_x)u\|^2$$
		applied with $\varepsilon:=\frac{\cos(\alpha/2)}{|\sin(\alpha/2)|}$ yields
		\begin{align*}\Re\langle e^{-i \frac{\alpha}{2}} P^\Phi(h)u,u\rangle &\geq \frac{1}{\cos(\alpha/2)}\int_\R \left[\cos^2(\alpha/2) U-\Phi'^2 \right]|u|^2\\
			&\geq \frac{Mh}{\cos(\alpha/2)} \int_\R (1-\chi^2)|u|^2\,.
		\end{align*}
		since $\Phi$ is an $(Mh)$-subsolution associated to $\chi$. Therefore, 
		\begin{align*}
		 \|u\|^2 &\leq \frac{\cos(\alpha/2)}{Mh}\big(\|(P^\Phi(h) - z)u\|_{L^2} + \|zu\|_{L^2}\big)\|u\|_{L^2} + \|\chi u\|^2_{L^2} \\
		 & \leq \frac{\cos(\alpha/2)^2}{2M^2} h^{-2} \|(P^\Phi - z)u\|^2 + \left( \frac{1}{2} + \frac{R}{M} \right) \|u\|^2 + \|\chi u\|^2,
		\end{align*}
		which implies \eqref{eq:sufficientAgmon} since $\frac12 + \frac RM < 1$.
	\end{proof}
	
	\subsection{Exponential localization of eigenfunctions of the simple-well operator}\label{sec:Agmon}
	
	\subsubsection{Simple-well operator}\label{s:simple_well0}
	We ``seal'' the potential $V$ at the right well, i.e., we consider
	$V_\ell:=V+\Sigma_\ell$,
	where $\Sigma_\ell : \R \to \R_+$ has compact support in $(x_r-\eta,x_r+\eta)$, $\eta > 0$ and satisfies $\Sigma_\ell(x_r) > 0$ (see Figure \ref{fig:potential} below). The parameter $\eta$ will remain fixed in the remainder of this article unless stated otherwise. 
	The resulting {\em simple-well operator} (or left-well operator) is
	$$\mathscr{L}_\ell(h) :=h^2 D_x^2+e^{i\alpha}V_\ell.$$
	This will be the main object of study in Sections \ref{sec:elliptic1D} to \ref{sec:proofWKB}, until we return to the double-well operator in Section \ref{sec:double_well1D}. 
	
	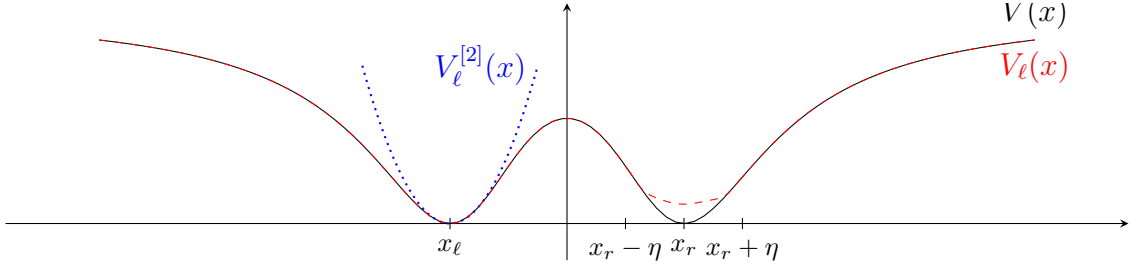
\begin{figure}[htbp]
		\centering
		\begin{tikzpicture}[domain=-8:8,samples=100,declare function={ V(\x)  =  4*(1 - (\x)^2)^2/(2 + (\x)^4);V2(\x) = 16/3*(\x+1)^2 ; bump(\x) =exp(-1/(1-(\x)^2)^2);}]
			\begin{axis}[
				width=\textwidth,
				height=5cm,
				axis x line=center,
				axis y line=center, 
				xtick = \empty,
				ytick = \empty,
				extra x ticks = {-2,1,2,3},
				extra x tick labels = {$x_\ell$,$x_r-\eta$,$x_r$,$x_r + \eta$},
				enlarge x limits={rel=0.1},
				enlarge y limits={rel=0.2},
				]
				\addplot[domain=-8:8] {V(\x/2)} node[pos=1] (endofplotV){};
				\node[above] at (endofplotV) {$V(x)$};
				\addplot[domain=-3.5:-0.5,color=blue,thick,dotted] {V2(\x/2)} node[left] {$V_\ell^{[2]}(x)$};
				\addplot[domain=1.05:2.95,color=red,dashed] {V(\x/2)+bump(\x - 2)};
				\addplot[domain=-8:1.05,color=red,dashed] {V(\x/2)};
				\addplot[domain=2.95:8,color=red,dashed] {V(\x/2)} node[below] {$V_\ell(x)$};
			\end{axis}
		\end{tikzpicture}
		\caption{Graph of a function $V$ satisfying the assumptions (black, solid line) and the potential $V_\ell = V + \Sigma_\ell$ (red, dashed line) for the simple well problem and $V_\ell^{[2]}$ the Taylor approximation of $V_\ell$ to order $2$ at $x_\ell$.}
		\label{fig:potential}
	\end{figure}

\subsubsection{Agmon distance and exponential localization}
In what follows, we fix $A > 0$ such that $[-A,A]$ contains $[x_\ell,x_r]$, and let $V_{\ell,A}$ be a smooth compactly supported function such that $0 \leq V_{\ell,A} \leq V_\ell$ and $V_{\ell,A}$ coincides with $V_{\ell}$ on the interval $[-A,A]$. Given $\varepsilon \in (0,1)$, define the {\em Agmon distance}
\begin{equation}
	\label{ex:Phieps}\Phi_\varepsilon(x):=  \sqrt{1-\varepsilon}\left|\int_{x_\ell}^x\cos(\alpha/2)\sqrt{V_{\ell,A}(s)}\,ds\right|.
\end{equation}
Note that thanks to the compact support of $V_{\ell,A}$, $\Phi_\varepsilon$ is bounded on $\R$. Furthermore, letting 
\begin{equation}
\label{eq:defSgamma}
S_-(\gamma) := \int_{x_\ell-\gamma}^{x_\ell} \cos(\alpha/2) \sqrt{V_{\ell,A}(s)}\,ds\,, \quad S_+(\gamma) := \int_{x_\ell}^{x_\ell + \gamma} \cos(\alpha/2) \sqrt{V_{\ell,A}(s)}\,ds,
\end{equation}
we have the following decay properties.
\begin{lemma}[Decay of the Agmon distance]\
\label{lem:decayPhi}
\begin{itemize}
\item[(i)] Given $\varepsilon > 0$, there exists $\gamma_0 > 0$ such that, for any $N \in \N$, there exists $C > 0$ such that
$$\|\mathbf{1}_{(x_\ell-\gamma_0,x_\ell+\gamma_0)}(x-x_\ell)^N e^{-{\Phi_\varepsilon}/h} \|_{L^\infty} \leq C h^{N/2} \qquad \textup{for all } h > 0.$$
\item[(ii)] Given $\delta > 0$, there exists $\varepsilon > 0$ such that, for all $\gamma$, $N_0 > 0$, and any $\widetilde{\chi}_\ell^{\pm} \in C^\infty(\R,[0,1])$ satisfying $$\supp (\widetilde{\chi}_\ell^-) \subset (-\infty,x_\ell - \gamma]\,, \quad \supp (\widetilde{\chi}_\ell^+) \subset [x_\ell + \gamma,+\infty)\,,$$
there exists $h_0$, $C > 0$ such that
$$\max_{0 \leq N \leq N_0} \|\widetilde{\chi}_{\ell}^{\pm}(x-x_\ell)^{N}e^{-\Phi_\varepsilon/h}\|_{W^{1,\infty}} \leq C e^{-(S_\pm(\gamma) - \delta)/h} \qquad \textup{for all } h \in (0,h_0)\,.$$
\end{itemize}
\end{lemma}
\rev{The proof, which we omit, relies on the fact that the weight $e^{-\Phi_\varepsilon/h}$ is Gaussian near $x_\ell$ and localized near $x_\ell$ at scale $O(\sqrt{h})$, and that $\Phi_\varepsilon \geq S_+(\gamma)$ on the support of $\widetilde{\chi}_+$.}

Let $\chi_0 \in C^\infty_c(\R)$, $0 \leq \chi_0 \leq 1$, be such that $\chi_0 \equiv 1$ on $[-1,1]$, $\supp \chi_0 \subset [-2,2]$, and for $L,\beta > 0$, let  
$$\chi_{L,h,\beta}(x) := \chi_0\left(\frac{x-x_\ell}{Lh^\beta}\right).$$
When $\beta = \frac12$, we simply write $\chi_{L,h} := \chi_{L,h,\frac12}$. We now show that $\Phi_\varepsilon$ is an $(Mh)$-subsolution associated to $\chi_{L,h}$ for the potential $V_\ell$ (in the sense of Definition \ref{def:M_solution}). Let us start by recording the following two lemmas, which will be also useful later on (the case $\beta \neq \frac12$ in Lemma \ref{lem:fatigue} is used in Lemma \ref{lem:elliptic_far} below). Their proof are elementary \rev{(using that $V_\ell$ and $\Phi$ behave quadratically near $x_\ell$)}. 
\begin{lemma}\label{lem:fatigue}
	For any positive constant $M > 0$, there exists $L > 0$ such that, for all $\beta > 0$, there exists $h_0 > 0$ such that
	$$V_\ell(x) \geq M h^{2\beta}(1 - \chi_{L,h,\beta}^2(x))\,, \quad x \in \R\,,\,\, h \in (0,h_0).$$
\end{lemma}


\begin{lemma}
	\label{lem:chih_subsolution}
	Let $M > 0$, let $L > 0$ and $h_0 > 0$. Then there exists $C > 0$ such that for all $h \in (0,h_0)$, if $\Phi$ is an $(Mh)$-subsolution associated to $\chi_{L,h}$ for the potential $V_\ell$, and if in addition, $\Phi(x_\ell) = 0$, then
	$$\|\chi_{L,h} e^{\Phi/h}\|_{\infty} \leq C.$$ 
\end{lemma}

\begin{lemma}[The Agmon distance is a subsolution]
	\label{lem:Phieps_Msubsolution}
	For all $\varepsilon \in (0,1)$ and $M > 0$, there exists $L > 0$ and $h_0 > 0$ such that for all $h \in (0,h_0)$, the function $\Phi_\varepsilon$ defined by \eqref{ex:Phieps} is an $(Mh)$-subsolution associated to $\chi_{L,h}$ for the potential $V_\ell$ in the sense of Definition \ref{def:M_solution}.
\end{lemma}
\begin{proof}
	Since $V_{\ell,A} \leq V_\ell$, one has
	$$\Phi_\varepsilon'^2 - \cos(\alpha/2)^2 V_\ell \leq -\varepsilon \cos(\alpha/2)^2 V_\ell,$$
	and the conclusion follows from Lemma \ref{lem:fatigue} applied with $\beta = \frac12$. 
\end{proof}


\begin{corollary}[Agmon estimate]
	\label{cor:Agmon1}
	Let $R > 0$, $\varepsilon > 0$, and let $\Phi_\varepsilon$ be defined by \eqref{ex:Phieps}. Then, there exists $C > 0$ and $h_0 > 0$ such that the estimate
	\begin{equation}
		\label{eq:LocExp}
		\|e^{\Phi_\varepsilon/h}\psi\|_{H^2_h}\leq C\|\psi\|_{L^2}.
	\end{equation}
	holds for all $h \in (0,h_0)$ and $\psi \in H^2(\R)$ satisfying $(\mathscr{L}_\ell(h) - \mu)\psi = 0$ with $\mu \in D(0,Rh)$. 
\end{corollary}
\begin{proof}
	Let $M > 2R$. By \rev{Lemma \ref{lem:Phieps_Msubsolution}}, we can choose $L,h_0 > 0$ such that $\Phi_\varepsilon$ is an $(Mh)$-subsolution associated to $\chi_{L,h}$ for all $h \in (0,h_0)$. Applying Proposition \ref{prop:estimee_elliptique_conj} with $\Phi := \Phi_\varepsilon$, $z := \mu$ and $u := e^{\Phi_\varepsilon/h}\psi$ (which indeed belongs to $H^2$ since $\Phi_\varepsilon$ is bounded on $\R$) we obtain
	\begin{align*}
		\|e^{\Phi_\varepsilon/h} \psi \|_{H^2_h} &\leq Ch^{-1} \|(\mathscr{L}^{\Phi_\varepsilon}(h) - \mu) e^{\Phi_\varepsilon/h} \psi \| + C\|\chi_{L,h} e^{\Phi_\varepsilon/h}\psi\|\,,
	\end{align*}
	where $\mathscr{L}^{\Phi_\varepsilon}(h) := e^{\Phi_\varepsilon/h} \mathscr{L}(h) e^{-\Phi_\varepsilon/h}$. Noticing that $(\mathscr{L}_\ell^{\Phi_\varepsilon}(h) -\mu) e^{\Phi_\varepsilon/h} \psi =  e^{\Phi_\varepsilon/h} (\mathscr{L}_\ell(h) - \mu)\psi = 0$, and using Lemma \ref{lem:chih_subsolution}, the result follows. 
\end{proof}
\begin{corollary}[Exponential localization of eigenfunctions near $x_\ell$]
		\label{cor:small_far}
		For every $R > 0$ \rev{and $N > 0$}, the following properties hold.  
		\begin{itemize}
		\item[(i)] There exists $h_0$, $\gamma_0$, $C > 0$ such that the inequality 
		$$\|\mathbf{1}_{(x_\ell - \gamma_0,x_\ell + \gamma_0)} (x-x_\ell)^{N} \psi\|_{L^2} \leq C h^{N/2} \|\psi\|_{L^2}$$
		holds for all $h \in (0,h_0)$ and $\psi \in H^2(\R)$ satisfying $(\mathscr{L}_\ell(h) - \mu) \psi = 0$ with $\mu\in D(0,Rh)$. 
		\item[(ii)] Given $\gamma > 0, N_0 > 0, \delta > 0$ and $\widetilde{\chi}_\ell^\pm \in C^\infty(\R,[0,1])$ such that
		$$\supp (\widetilde{\chi}_\ell^-) \subset (-\infty,x_\ell - \gamma]\,, \quad \supp (\widetilde{\chi}_\ell^+) \subset [x_\ell + \gamma,+\infty)\,,$$
		there exists $h_0 > 0$ and $C > 0$ such that the inequality
		\begin{equation}\label{eq:loc vecteur propre}
			\max_{0 \leq N \leq N_0}\left\|\widetilde{\chi}_\ell^{\pm}(x-x_\ell)^N\psi\right\|_{H^1_h}\leq Ce^{-(S_\pm(\gamma)-\delta)/h}\|\psi\|_{L^2},
		\end{equation}
		with $S_\pm(\gamma)$ defined by \eqref{eq:defSgamma}, 
		holds for all $h \in (0,h_0)$ and all $\psi \in L^2(\mathbb{R})$ satisfying $(\mathscr{L}_\ell(h) - \mu)\psi = 0$ with $\mu \in D(0,Rh)$. 
		\end{itemize}

	\end{corollary}
	
	\begin{proof}
		Let $\varepsilon > 0$ be small enough and let $\Phi_\varepsilon$ be defined as in \eqref{ex:Phieps} where $A$ is chosen large enough. Then
		$$
		\|\mathbf{1}_{(x_\ell-\gamma_0,x_\ell+\gamma_0)} (x-x_\ell)^N \psi\|_{L^2} \leq \|\mathbf{1}_{(x_\ell-\gamma_0,x_\ell+\gamma_0)} (x-x_\ell)^N e^{-\Phi_\varepsilon/h}\|_{L^\infty} \cdot \|e^{\Phi_\varepsilon/h} \psi\|_{L^2}\,,
		$$
		$$
		\left\|\widetilde{\chi}^\pm_{\ell}(x-x_\ell)^N\psi\right\|_{H^1_h}\leq C\left\|\widetilde{\chi}^+_{\ell}(x-x_\ell)^N e^{-\Phi_\varepsilon/h} \right\|_{W^{1,\infty}}\|e^{\Phi_\varepsilon/h} \psi\|_{H^1_h}
		$$
		and the conclusion follows immediately from Corollary \ref{cor:Agmon1} and Lemma \ref{lem:decayPhi}.
	\end{proof}
	
	\subsubsection{Quasimodes of the double-well operator and its quadratic approximation}
	Let $V_\ell^{[2]}$ be the quadratic approximation of $V_\ell$ at $x = x_\ell$, i.e., $V_\ell^{[2]} := \frac{V''(x_\ell)}{2}(x-x_\ell)^2$, and let
	\begin{equation}
		\label{eq:defL[2]}
		\mathscr{L}_{\ell}^{[2]}(h) := (hD_x)^2 + e^{i\alpha} V_\ell^{[2]}.
	\end{equation}
	Moreover, define
	\begin{equation}
		\label{eq:def_Seta}
		S_\eta := \int_{x_\ell}^{x_r - \eta} \cos(\alpha/2) \sqrt{V(x)}\,dx
	\end{equation}
	(Recall that $\eta$ is the parameter chosen in Section \ref{s:simple_well0} such that the support of $\Sigma_\ell = V - V_\ell$ lies in $(x_r - \eta,x_r + \eta)$, see also Figure \ref{fig:chiell} below).  
	
	\begin{corollary}[Quasimodes of $\mathscr{L}^{[2]}_\ell(h)$ and $\mathscr{L}(h)$]
		\label{cor:quasimodes}
		Let $\chi_\ell \in C^\infty(\R,[0,1])$ be such that $\chi_\ell \equiv 1$ on $(-\infty,x_r - \eta]$ and $\chi_\ell \perp \Sigma_\ell$, and let $R > 0$, $\delta > 0$. Then, there exists $C,h_0 > 0$ such that the inequalities
		\begin{equation}
			\label{eq:quasimode_exp}
			\|(\mathscr{L}(h)-\mu)\psi\| + \|(\mathscr{L}(h)-\mu)(\chi_\ell\psi)\|\leq C e^{-(S_\eta-\delta)/h}\|\psi\|\,,
		\end{equation}
		(where $\mathscr{L}(h)$ is the original, double-well operator), and
		\begin{equation}
			\label{eq:quasimode_h32}
			\|(\mathscr{L}^{[2]}_\ell(h)-\mu)\psi\| \leq C h^{3/2}\|\psi\|\,,
		\end{equation}
		hold for all $h \in (0,h_0)$ and all $\psi \in H^2(\R)$ satisfying $(\mathscr{L}_\ell(h) - \mu)\psi = 0$ with $\mu \in D(0,Rh)$. 
	\end{corollary}
	
	\begin{proof}
		Let $\gamma := x_r - \eta$, and denote $\widetilde{\chi}_\ell := 1 - \chi_\ell$. Then $\supp \widetilde{\chi}_\ell \subset (x_r - \eta,+\infty)$. Let $\underline{\widetilde{\chi}}_\ell$ be such that $\widetilde{\chi}_\ell \prec \underline{\widetilde{\chi}}_\ell$ and $\supp(\underline{\chi}_\ell) \subset (x_r - \eta,+\infty)$. These functions are sketched in Figure \ref{fig:chiell} below.
		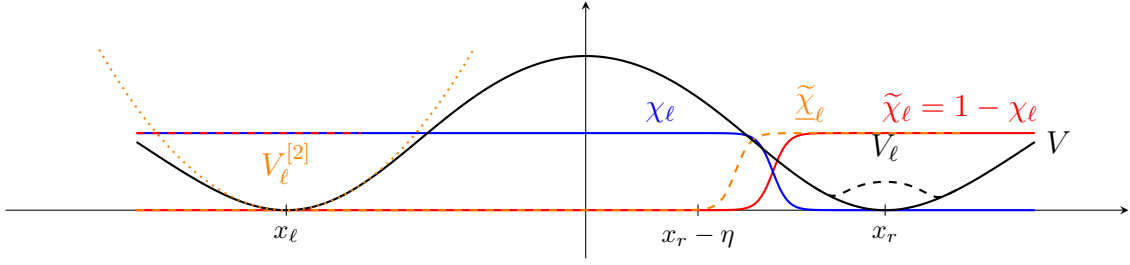
\begin{figure}[H]
			\centering
			\begin{tikzpicture}[domain=-6:7,samples=300,declare function={plateau(\x) = 0.5*(1+tanh(5*(\x))); V(\x)  =  4*(1 - (\x)^2)^2/(2 + (\x)^4);  bump(\x) =exp(-1/(1-(\x)^2)^2); V2(\x) = 16/3*(\x+1)^2 ;}]
				\begin{axis}[
					width=\textwidth,
					height=5cm,
					axis x line=center,
					axis y line=center, 
					xtick = \empty,
					ytick = \empty,
					extra x ticks = {-4,1.5,4},
					extra x tick labels = {$x_\ell$,$x_r-\eta$,$x_r$},
					enlarge x limits={rel=0.1},
					enlarge y limits={rel=0.3},
					]
					\addplot[domain=-6:6,color=black,thick,solid,color=red] {plateau((\x)-2.5)};
					\addplot[domain=-6:6,color=black,thick,solid,color=blue] {1 - plateau((\x)-2.5)};
					\addplot[domain=-6:5,thick,solid,color=orange,dashed] {plateau((\x)-2)};
					\addplot[domain=-6:6,color=black,thick,solid,color=black] {V((\x)/4)} node[right] {$V$};
					\addplot[domain=-6.5:-1.5,color=black,thick,dotted,color=orange] {V2((\x)/4)};
					\node[above,color=orange] at (-4,0.2) {$V_\ell^{[2]}$};
					\addplot[domain=3.05:4.95,color=black,thick,dashed,color=black] {V((\x)/4) + bump(2*((\x)/2 - 2))};
					\node[above] at (4,0.5) {$V_\ell$};
					\node[above,color=red] at (5,1) {$\widetilde{\chi}_\ell = 1 - \chi_\ell$};
					\node[above,color=orange] at (3,1) {$\underline{\widetilde{\chi}}_\ell$};
					\node[above,color=blue] at (1,1) {$\chi_\ell$};
					\addplot[domain=-6:-3,color=red,dashed,thick]{(\x)/(\x)};
				\end{axis}
			\end{tikzpicture}
			\caption{The cutoff functions $\tilde{\chi}_\ell$ and $\underline{\tilde{\chi}_\ell}$ in Corollary \ref{cor:quasimodes}.}
			\label{fig:chiell}
		\end{figure}%
		\noindent Since $\chi_\ell \perp \Sigma_\ell$ and since $(\mathscr{L}_\ell(h) - \mu) \psi = 0$, we have
		\begin{align*}
			(\mathscr{L}(h) - \mu)(\chi_\ell \psi)  
			& = -[\mathscr{L}_\ell(h),\widetilde{\chi}_\ell] (\underline{\widetilde{\chi}}_\ell \psi).
		\end{align*} 
		Thus, by the property (ii) in Corollary \ref{cor:small_far},
		$$\|(\mathscr{L}(h)-\mu)(\chi_\ell \psi)\| \leq C \|\underline{\widetilde{\chi}}_\ell \rev{\psi} \|_{H^1_h} \leq C e^{-(S_\eta-\delta)/h}\|\psi\|,$$
		and also, since $(\mathscr{L} - \mu)\psi = \Sigma_\ell \psi = (1-\chi_\ell)\Sigma_\ell \psi$, 
		$$\|(\mathscr{L}(h) - \mu)\psi\| \leq Ce^{-(S_\eta - \delta)/h} \|\psi\|.$$
		The previous two estimates give \eqref{eq:quasimode_exp}.  On the other hand, let $\gamma_0 > 0$ be small enough and let
		$\chi_{\ell,\gamma_0} \in C^\infty_c(\R)$ be such that $\chi_{\ell,\gamma_0} \equiv 1$ on $[-\gamma_0/2,\gamma_0/2]$ and $\supp \chi_{\ell,\gamma_0} \subset (x_\ell-\gamma_0,x_\ell+\gamma_0)$. Let $\widetilde{\chi}_{\ell,\gamma_0}:= 1 - \chi_{\ell,\gamma_0}$. Then
		\begin{align}
			(\mathscr{L}^{[2]}_\ell(h)-\mu)\psi
			=e^{i\alpha}\chi_{\ell,\gamma_0}(\rev{V_\ell^{[2]}-V_\ell})\psi+e^{i\alpha}\widetilde{\chi}_{\ell,\gamma_0}(\rev{V_\ell^{[2]}-V_\ell})\psi.
			\label{eq:quasimode_NSA_2parts}
		\end{align} 
		The first term in the right-hand side can be estimated via
		\begin{align*}
			\|e^{i\alpha}\chi_{\ell,\gamma_0}(\rev{V_\ell^{[2]}-V_\ell})\psi\|
			& \leq C\|V_\ell^{(3)}\|_{\infty} \| \mathbf{1}_{(x_\ell-\gamma_0,x_\ell+\gamma_0)}(x - x_\ell)^3 \psi\|_{L^2} \leq C h^{3/2}\|\psi\|_{L^2},
		\end{align*}
		by the property (i) in Corollary \ref{cor:small_far}, and the second, via
		\begin{align*}
			\|e^{i\alpha}\rev{\widetilde{\chi}_{\ell,\gamma_0}}(\rev{V_\ell^{[2]}-V_\ell})\psi\|
			& \leq \|V \widetilde{\chi}_{\ell,\gamma_0} \psi\| + \|V_\ell^{[2]}\widetilde{\chi}_{\ell,\gamma_0} \psi\|\\
			&\leq C\sup_{N \leq 2}\|\widetilde{\chi}_{\ell,\gamma_0}(x-x_\ell)^N\psi\| \leq C e^{-s/h}\|\psi\|_{L^2}
		\end{align*}
		for some constant $s > 0$, by the property (ii) in Corollary \ref{cor:small_far} (applied, e.g., with $\gamma = \gamma_0/4$). Using these two estimates in \eqref{eq:quasimode_NSA_2parts} yields \eqref{eq:quasimode_h32}, concluding the proof.		
	\end{proof}
	
	\section{Resolvent of the simple-well operator}\label{sec:3}

	In this section, we study the resolvent $z \mapsto (\mathscr{L}_\ell(h) - z)^{-1}$ for $z \in D(0,Rh)$. Its poles -- the eigenvalues of $\mathscr{L}_\ell(h)$ -- are found asymptotically as $h \to 0$, and the behaviour of $(\mathscr{L}_\ell(h) - z)^{-1}$ near these poles is described with the help of the associated Riesz projectors, see Proposition \ref{prop:fourre_tout}. These results will allow us to ``bypass'' the spectral theorem in the subsequent analysis. To establish these properties, a key role is played by the complex harmonic oscillator, as can be guessed from Corollary \ref{cor:quasimodes} above. We thus begin by some background on this operator.

	\subsection{The complex harmonic oscillator}
	\label{sec:davies}
	Given $h > 0$ and $\alpha \in (-\pi,\pi)$, we consider the complex harmonic oscillator
	\[\mathscr{H}(h)=(hD_x)^2+e^{i\alpha}x^2\,,\]
	see, e.g., \cite{davies1982pseudo}, \cite[Section 14.5]{davies2007linear}, \cite[Section 14.4]{helffer2013spectral} and \cite[section 7.4]{krejcirik2015pseudospectra}. Denote by 
	$$\textup{dom}(\mathscr{H}(h)) := \{u \in H^2(\mathbb{R}) \, \mid\, xu \in H^1(\R) \textup{ and } x^2 u \in L^2(\R)\}$$
	its domain, see \cite[Eq.~14.4.3]{helffer2013spectral}. The point is that the operator $\mathscr{L}_\ell^{[2]}(h)$ of the previous paragraph is unitarily equivalent to $\mathscr{H}(ah)$. More precisely, 
	$$\mathcal{U}^* \mathscr{L}_\ell^{[2]} \mathcal{U}  =\mathscr{H}(ah)\,,$$
	$$\mathcal{U}f(x) = \sqrt{a}f(a(x-x_\ell)) \quad \textup{where}\quad a := \sqrt{V''(x_\ell)/2}\,^.$$
	Recall that $\mathscr{H}(h)$ has a purely discrete spectrum consisting of algebraically simple eigenvalues which are given by 
	\begin{equation}
		\label{eq:defnuh}
		\nu_n(h) := \displaystyle(2n-1)e^{i\frac{\alpha}{2}} h, \quad n = 1,2,...
	\end{equation} 
	\begin{proposition}[Resolvent bound for the complex harmonic operator]
		\label{prop:resboundH}
		For all $R>0$, there exists $C>0$ such that, for all $h > 0$ and $z\in D(0,Rh) \setminus \textup{sp}\big(\mathscr{H}(h)\big)$,
		\[\|(\mathscr{H}(h)-z)^{-1}\| \leq \frac{C}{\dist\big(z,\mathrm{sp}(\mathscr{H}(h))\big)}\,.\]
	\end{proposition}
	\begin{proof}
	We have
	$$\mathscr{H}(h) - z= h \mathcal{U}^{-1} (\mathscr{H}(1) - h^{-1}z) \mathcal{U}$$
	where $\mathcal{U}$ is the isometry of $L^2$ defined by
	$$\mathcal{U} v := h^{1/2} v(h^{1/2}x).$$
	The conclusion then follows from the fact that $\textup{dist}(h^{-1}z,\textup{sp}(\mathscr{H}(1))) = h^{-1} \textup{dist}(z,h\,\textup{sp}(\mathscr{H}(1))) = h^{-1} \textup{dist}(z,\textup{sp}(\mathscr{H}(h)))$ together with the bound
	$$\|(\mathscr{H}(1) - z)^{-1}\| \leq \frac{C_1 e^{C_2 |z|}}{\textup{dist}(z,\textup{sp}(\mathscr{H}(1)))}$$
	which can be found in  \cite[section 7.4]{krejcirik2015pseudospectra}. 
	\end{proof}
	
	\begin{remark}
	We highlight that the resolvent estimate of Proposition \ref{prop:resboundH} famously fails if $z$ is allowed to be arbitrary in $\mathbb{C} \setminus \textup{sp}\big(\mathscr{H}(h)\big)$. This follows from the pseudospectral properties of $\mathscr{H}(h)$ analyzed, e.g., in \cite{davies1982pseudo}.
	\end{remark}
	\begin{proposition}
		\label{prop:resboundDx}
		Let $R> 0$. There exists $C > 0$ such that for all $z \in D(0,Rh) \setminus \textup{sp}\big(\mathscr{H}(h)\big)$, \begin{align*}
			&\|(\mathscr{H}(h) - z)^{-1}(hD_x)\| + \|(\mathscr{H}(h) - z)^{-1}x\| + \|(hD_x)(\mathscr{H}(h) - z)^{-1}\| + \|x (\mathscr{H}(h) - z)^{-1}\|\\
			&\hspace{5cm} \leq C \left(\frac{1}{\sqrt{\textup{dist}(z,\textup{sp}(\mathscr{H}(h)))}} + \frac{\sqrt{h}}{\textup{dist}(z,\textup{sp}(\mathscr{H}(h)))}\right).
		\end{align*} 
	\end{proposition}
	\begin{proof}
		Let $z \in D(0,Rh)$ and let $d := \textup{dist}(z,\textup{sp}(\mathscr{H}(h)))$. For all $\psi \in \textup{dom}(\mathscr{H}(h))$, 
		$$\textup{Re} \langle e^{-i\alpha/2}(\mathscr{H}(h) - z) \psi,\psi \rangle \rev{\,\leq\,} \cos(\alpha/2)(\|(hD_x) \psi\|^2 + \|x \psi\|^2) - Rh \|\psi\|^2.$$
		Thus, by Proposition \ref{prop:resboundH} above, 
		\begin{align*}
			\cos(\alpha/2) \left(\|(hD_x) \psi\|^2 + \|x\psi\|^2\right)
			&\leq \|\psi\| \|(\mathscr{H}(h) - z)\psi\| + Rh \|\psi\|^2\\
			&\leq C \left(d^{-1} + hd^{-2}\right) \|(\mathscr{H}(h) - z)\psi\|^2,
		\end{align*}
		which gives the bound for the terms
		$\|(hD_x)(\mathscr{H}(h) - z)^{-1} \|$ and $\|x(\mathscr{H}(h) - z)^{-1}\|$.  The other two estimates follow from these by duality, after switching $\alpha \to -\alpha$. 	
	\end{proof}
	
	Finally, we will need the following pseudo-local behaviour of $(\mathscr{H}(h) - z)^{-1}$ at scales $h^{\beta}$ for $0 \leq \beta < \frac12$ for small complex numbers $z$ at a safe distance from the spectrum.\footnote{This result will be used only with $\tau = 1$, where $\tau$ is the parameter appearing in the statement of the proposition, but we state the more general version here since it does not significantly complicate the proof.}
	\begin{proposition}[Pseudo-locality of $(\mathscr{H}(h) - z)^{-1}$]
		\label{prop:pseudolocH}
		Let $R,c, N > 0$, $\beta \in [0,\frac12)$, 
		$\tau < \frac32-\beta$, and $\chi, \widetilde{\chi} \in C^\infty(\R)$ with $\chi \perp \widetilde{\chi}$. Then there exists $C,h_0 > 0$ such that the inequality
		$$\|\chi_h (\mathscr{H}(h) - z)^{-1} \widetilde{\chi}_h\| \leq C h^{N}$$
		holds for all $h \in (0,h_0)$ and $z \in D(0,Rh)$ satisfying $\dist\big(z,\textup{sp}(\mathscr{H}(h))\big) \geq ch^{\tau}$, with $\chi_h = \chi(h^{-\beta}x)$ and $\widetilde{\chi}_h =\widetilde{\chi}(h^{-\beta}x)$. 
	\end{proposition}
	
	\begin{proof}
		Let $v \in L^2(\R)$ and let $u$ solve
		\begin{equation}
			\label{eq:aux_uv}
			(\mathscr{H}(h) - z) u = \widetilde{\chi}_h v.
		\end{equation}
		Multiplying \eqref{eq:aux_uv} by $\chi_h$ and commuting $\mathscr{H}(h)$ and $\chi_h$, we find
		$$(\mathscr{H}(h) - z) (\chi_h u) = [(hD_x)^2,\chi_h] (\underline{\chi_h}u)$$
		where $\underline{\chi_h}(x) := \underline{\chi_0}(h^{-\beta}x)$ for some $\underline{\chi} \in C^\infty(\R)$ with $\chi \prec \underline{\chi}$.
		Applying the resolvent estimates of Propositions \ref{prop:resboundH} and \ref{prop:resboundDx}, we deduce that
		\begin{align*}
			\|\chi_h u\| 
			&\leq \left(h^{2-2\beta} \|D_x^2 \chi\|_{\infty} \|(\mathscr{H}(h) - z)^{-1}\| + h^{1-\beta} \|D_x\chi\|_\infty \|(\mathscr{H}(h) - z)^{-1}(hD_x)\|\right) \|\underline{\chi_h}u\|.\\
			& \leq C (h^{2-2\beta - \tau} + h^{1 - \beta-\tau/2} + h^{\frac32 - \beta -\tau})\|\underline{\chi_h}u\|\\
			& \leq C h^{\varepsilon} \|\underline{\chi_h}u\|,
		\end{align*}
		for some $\varepsilon > 0$, since $\tau < \frac32 - \beta$ and $\beta < \frac12$. 
		Repeating this argument for a sequence of functions $\underline{\chi_0} \prec \chi_1 \prec \chi_2 \prec \ldots \prec \chi_M$ with $\chi_M \perp \widetilde{\chi}$, and using Proposition \ref{prop:resboundH}, we conclude that 
		$$\|\chi_h u\| \leq C h^{M\varepsilon} \|\chi_{M,h} u\| \leq C h^{\varepsilon M-\tau} \|\widetilde{\chi}_h v\| \leq C h^{N} \|v\|$$
		where $\chi_{M,h}(x) = \chi_M(h^{-\beta}x)$ and  $N = \varepsilon M-\tau$ can be made arbitrarily large.
	\end{proof}

	\subsection{Nature and a priori location of the spectrum}
	\label{sec:nature_a_priori}
	We now locate asymptotically the poles of $(\mathscr{L}_\ell - z)^{-1}$ in $D(0,Rh)$.
	\begin{proposition}
		\label{prop:point_spectrum}
		For all $R > 0$, there exists $h_0 > 0$ such that for all $h \in (0,h_0)$, the spectrum of $\mathscr{L}_\ell(h)$ in $D(0,Rh)$ is purely discrete.
	\end{proposition}
	\begin{proof}
		We show that there exists $r > 0$ such that for all $z\in D(0,r)$, $\mathscr{L}_\ell(h) - z$ is Fredholm of index $0$. The conclusion then follows by the analytic Fredholm theory (see e.g. the argument in the proof of \cite[Theorem 5.11]{cheverry2021guide}) and taking $h_0$ small enough.
		
		Let $V_\infty := \lim \inf_{x \to \pm\infty}V_{\ell}(x) > 0$ and let $z \in D(0,V_\infty/3)$.
		Since $K := \{x \in \R \,\mid\, V(x) - \frac{V_\infty}{3} \leq \frac {V_\infty}{3}\}$ is compact, there exists
		$\chi \in C^\infty_c(\R)$ with $\chi \geq 0$ and $\chi \equiv 2\frac{V_\infty}{3}$
		on $K$. We then have
		$V + \chi - |z| \geq \frac{V_\infty}{3}$ on $\R$, and thus
		$$\Re\big\langle e^{-i\alpha/2}\big(\mathscr{L}_\ell(h) - z + \chi\big)u,u\big\rangle \geq \cos(\alpha/2) \left(\|(hD_x)u\|^2 + \frac{V_\infty}{3} \|u\|^2\right).$$
		It follows that 
		$\|(\mathscr{L}_\ell(h) + \chi - z) u\| \geq c\|u\|$
		for some $c > 0$, which implies that $\mathscr{L}_\ell(h) - z + \chi$ is injective with closed range \cite[Proposition 2.14]{cheverry2021guide}. Switching $\alpha$ to $-\alpha$ and $z$ to $\overline{z}$, the same argument shows that $(\mathscr{L}_\ell(h) - z + \chi)^*$ is injective with closed range; thus, $\mathscr{L}_\ell(h) - z + \chi$ is an isomorphism. Since $u \mapsto \chi u$ is compact from $H^2(\R) \to L^2(\R)$, we conclude that $\mathscr{L}_{\ell}(h) - z$ is Fredholm of index $0$. 
	\end{proof}

	\begin{proposition}
		\label{prop:locspec}
		Let $R > 0$. There exist $C,h_0>0$ such that, for all $h\in(0,h_0)$,
		\[\mathrm{sp}(\mathscr{L}_\ell(h)) \cap D(0,Rh)\subset\bigcup_{n=1}^{N(R)} D(\nu_n (ah),Ch^{\frac32})\,,\]
		where $N(R) :=  \lfloor \frac{R+a}{2a}\rfloor$, $\nu_n(h)$ are defined by \eqref{eq:defnuh}, and where we recall that $a = \sqrt{V''(x_\ell)/2}$.
	\end{proposition}
	\begin{proof}
		First observe that for all $h_0 > 0$ and for all $h \in (0,h_0)$, 
		$$\textup{sp}(\mathscr{H}(ah)) \cap D(0,Rh) = \{\nu_n(ah)\}_{1 \leq n \leq N(R)}.$$
		Let $\mu \in \textup{sp}(\mathscr{L}_\ell) \cap D(0,Rh)$ and, by Proposition \ref{prop:point_spectrum}, let $\psi$ be an associated eigenvector with unit norm. Using the resolvent bound for $\mathscr{H}(h)$ in Proposition \ref{prop:resboundH} and the fact that $\psi$ is a $\mathscr{O}(h^{3/2})$ quasimode of $\mathscr{H}(h)$ (Corollary \ref{cor:quasimodes}),  we get
		\begin{align*}
			d(\mu,\textup{sp}(\mathscr{H}(ah)))& \leq C
			\|(\mathscr{H}(ah) - \mu) \mathcal{U}^* \psi\|\\
			& = C\big\|\mathcal{U}^* \big((\mathscr{L}^{[2]}_\ell(h) - \mu)\psi \big)\big\| \leq Ch^{3/2}.
		\end{align*}
		Thus, there exists $n \in \{1,\ldots,N(R)\}$ such that $|\mu - \nu_n(ah) | \leq C h^{3/2}$.
	\end{proof}
	
	\subsection{Resolvent bounds for $\mathscr{L}_\ell(h)$} 
\label{sec:spectrum}
We are now in a position to define Riesz projectors for each disk $D(\nu_n(ah),Ch^{3/2})$ and state the main result of this section, Proposition \ref{prop:fourre_tout}.
 
 Namely, Proposition \ref{prop:locspec} implies that for all $R > 0$\rev{\footnote{Il suffirait de dire pour tout $n$}}, there exists $\varepsilon > 0$ and $h_0 > 0$ small enough such that, for all $h \in (0,h_0)$, the following Riesz projectors 
	\begin{equation}
		\label{eq:defPiln}
		\Pi_{\ell,n}(h)=\frac{1}{2i\pi}\int_{\mathscr{C}(\nu_n(ah),\epsilon h)}(z-\mathscr{L}_\ell(h))^{-1}dz\,,
	\end{equation}
	\begin{equation}
		\label{eq:defPi[2]ln}
		\Pi_{\ell,n}^{[2]}(h)=\frac{1}{2i\pi}\int_{\mathscr{C}(\nu_n(ah),\epsilon h)}(z-\mathscr{L}^{[2]}_\ell(h))^{-1}dz\,,
	\end{equation}
	are well defined for $n \in \{1,\ldots,N(R)\}$. The definition does not depend on the choice of $\varepsilon > 0$ small enough. Since the eigenvalues of $\mathscr{L}^{[2]}_\ell(h)$ are algebraically simple, $\textup{rank}(\Pi_{\ell,n}^{[2]}(h)) = 1$. 	
	\begin{proposition}
	\label{prop:fourre_tout}
	Let $R > 0$. There exists $C$, $h_0$, $\varepsilon > 0$ such that the following properties hold for all $n \in \{1,\ldots,N(R)\}$ and $h \in (0,h_0)$:
	\begin{itemize}
	\item[(i)] $\textup{rank}\big(\Pi_{\ell,n}(h)\big) = 1$ 
	\item[(ii)] $\|\Pi_{\ell,n}(h)\| \leq C$,
	\item[(iii)] There exists a unique eigenvalue of $\mathscr{L}_\ell(h)$ lying in $\mathscr{C}(\nu_n(ah),Ch^{3/2})$, denoted by $\mu_{\ell,n}(h)$. It is algebraically simple, and
	\begin{equation}
		\label{eq:Pi_vrai_projecteur}
		\textup{Ran}(\Pi_{\ell,n}(h)) =\textup{Ker}\big(\mathscr{L}_\ell(h) - \mu_{\ell,n}(h)\big).
	\end{equation}
	Moreover, for any \rev{non-zero} $\psi \in \textup{Ran}(\Pi_{\ell,n}(h))$, one has $\langle \psi,\overline{\psi} \rangle \neq 0$ and
	$$\Pi_{\ell,n}(h) = \frac{\langle \cdot,\overline{\psi}\rangle }{\langle \psi,\overline{\psi}\rangle} \psi.$$
	\item[(iv)] For all $u \in H^2(\R)$ and $z \in D(\nu_n(ah),\varepsilon h)$,
	$$\|(\textup{Id} - \Pi_{\ell,n}(h)) u\| \leq Ch^{-1} \big\|\big(\mathscr{L}_\ell(h) - z\big) u\big\|$$,
	\item[(v)] For all $u \in H^2(\R)  \setminus \{0\}$ and $z \in D(\nu_n(ah),\varepsilon h)$,
	$$|\mu_{\ell,n}(h) - z| \leq \frac{\|(\mathscr{L}_\ell(h) - z)u\|}{\|u\|}.$$
	\end{itemize} 
	\end{proposition}
	The proof can be found in \S\ref{sec:proof_2_props}. It relies on an estimate for the difference 
	\begin{equation}
	\label{eq:defDh(z)}
	D_h(z) := (\mathscr{L}_\ell(h)-z)^{-1}-(\mathscr{L}_\ell^{[2]}(h)-z)^{-1},
	\end{equation}
	established in \S\ref{sec:distance_resolvents} below.
%
%
%
%
%
%
%
 		
 	\subsubsection{Distance to the harmonic resolvent}
 	\label{sec:distance_resolvents}
 	We estimate $D_h(z)$ for $z$ at a safe distance of the spectrum of $\mathscr{L}_\ell^{[2]}(z)$, by decomposing into the region far from the well $x_\ell$, where both operators are elliptic, and close to $x_\ell$, where the operators are close to each other. 
  	
	\begin{lemma}[Ellipticity away from the well]
		\label{lem:elliptic_far}
		Let $R$, $N > 0$ and let $\beta\in [0,1/2)$. Let $\chi_0 \in C^\infty(\mathbb{R})$ be such that $\chi_0 \equiv 1$ in a neighborhood of $0$ and denote $\chi_h := \chi_0(h^{-\beta}(x-x_\ell))$. There exists $C$, $h_0>0$ such that the inequality
		\begin{equation}
			\label{eq:ellip_farL}
			\|(1-\chi_h)(\mathscr{L}_\ell(h)-z)^{-1}v\|\leq  Ch^{-2\beta} \|(1-\chi_h) v\|+ Ch^N\|(\mathscr{L}_\ell(h)-z)^{-1}\| \,\|v\|.
		\end{equation}
		holds for all $h\in(0,h_0)$, $z\in D(0,Rh) \setminus \textup{sp}(\mathscr{L}_\ell(h))$ and $v \in L^2(\mathbb{R})$, while 
		\begin{equation}
			\label{eq:ellip_farL[2]}
			\|(1-\chi_h)(\mathscr{L}^{[2]}_\ell(h)-z)^{-1}v\|\leq  Ch^{-2\beta} \|(1-\chi_h) v\|+ Ch^N\|(\mathscr{L}^{[2]}_\ell(h)-z)^{-1}\| \,\|v\|\,,
		\end{equation}
		holds for all $h \in (0,h_0)$, $z \in D(0,Rh) \setminus\textup{sp}(\mathscr{L}_\ell^{[2]}(h))$ and $v \in L^2(\R)$,   
	\end{lemma}
	\begin{proof}
		The proofs of \eqref{eq:ellip_farL} and \eqref{eq:ellip_farL[2]} being identical, we only prove the former. Write $\widetilde{\chi}_0 := 1 - \chi_0$ and let $\underline{\widetilde{\chi}_0}$ be chosen such that $\widetilde{\chi}_0 \prec \underline{\widetilde{\chi}_0}$ and $0 \notin \supp(\underline{\widetilde{\chi}_0})$. Let 
		$\widetilde{\chi}_h := 1 - \chi_h$ and 
		$$\underline{\widetilde{\chi}_h} := \underline{\widetilde{\chi}_0}\left((x - x_\ell)/{h^\beta}\right)\,.$$
		Given $v \in L^2(\mathbb{R})$, let $u_1,u_2$ be solutions of
		$$(\mathscr{L}_\ell(h)-z)u_1 =  \underline{\widetilde{\chi}_h} v\,, \quad (\mathscr{L}_\ell(h) - z)u_2 = (1 - \underline{\widetilde{\chi}_h}) v.$$
		After summing up and inverting, we get
		$$\|\widetilde{\chi}_h (\mathscr{L}-z)^{-1}v\| = \|\widetilde{\chi}_h (u_1 + u_2)\| \leq \|\widetilde{\chi}_h u_1\|_{L^2} + \|\widetilde{\chi}_h u_2\|_{L^2},$$
		and thus it suffices to show that
		\begin{equation}
		\label{eq:suffice_ellip_far}
		\|\widetilde{\chi}_h u_1\|_{L^2} \leq Ch^{-2\beta} \|\widetilde{\chi}_h v\|_{L^2} + Ch^N \|u_1\|_{L^2}\quad \text{and} \quad \|\widetilde{\chi}_h u_2\|_{L^2} \leq Ch^N \|u_2\|_{L^2}\,.
		\end{equation}
		
		The idea is that $\mathscr{L}_\ell(h) - z$ can be made elliptic by adding a perturbation localized away from the support of $\widetilde{\chi}_h$. Namely, let $\rho_0$ be chosen such that 
		$\rho_0 \perp \widetilde{\chi}_0$, $\rho_0 \equiv 1$ near $0$, and let
		$$\rho_h := \rho_0\left((x - x_\ell)/h^{\beta}\right).$$ 
		Notice that, by Lemma \ref{lem:fatigue} and since $\beta < \frac12$, we have $V_\ell + h^{2\beta} \rho_h-Rh \geq ch^{2\beta}$ on $\mathbb{R}$ for $h$ small enough.
		Thus, 
		\begin{align}
			\textup{Re}\big\langle e^{-i\alpha/2} \left(\mathscr{L}_\ell +  h^{2\beta}\rho_h - z \right) u,u\big\rangle 
			&\geq \cos(\alpha/2) \left( \|(hD_x) u\|^2_{L^2} + ch^{2\beta}\|u\|^2_{L^2} \right)
			\label{eq:ellip_far}
		\end{align}
	 	To exploit this property, we observe that since $\rho_h \widetilde{\chi}_h = 0$, 
		$$\widetilde{\chi}_h(\mathscr{L}_\ell + h^{2\beta}\rho_h - z) u_1 = \widetilde{\chi}_h v\,; 
		$$
		therefore, 
		$$(\mathscr{L}_\ell + h^{2\beta}\rho_h - z)(\widetilde{\chi}_h u_1) = \widetilde{\chi}_h v + [h^2 D_x^2,\widetilde{\chi}_h](\underline {\widetilde{\chi}_h} u_1).$$
		Testing by $\widetilde{\chi}_h u_1$, applying \eqref{eq:ellip_far} and using that $[(h D_x)^2,\widetilde{\chi}_h]$ is anti-symmetric in the $L^2$ scalar product, we deduce that 
		\begin{align*}
			&\cos(\alpha/2) \left(\|(hD_x) (\widetilde{\chi}_h u_1)\|^2_{L^2} + ch^{2\beta} \|\widetilde{\chi}_h u_1\|^2_{L^2}\right) \\
			&\quad\leq \left|\big\langle\widetilde{\chi}_h v,\widetilde{\chi}_h u_1\big\rangle\right| + \Big|\big\langle\underline {\widetilde{\chi}_h}u_1,[(h D_x)^2,\widetilde{\chi}_h](\widetilde{\chi}_h u_1)\big\rangle\Big|\\
			&\quad \leq \|\widetilde{\chi}_h v\|_{L^2} \|\widetilde{\chi}_h u_1\|_{L^2} + C\|\underline {\widetilde{\chi}_h}u_1\|_{L^2} \Big(h^{1 - \beta} \|(hD_x)\widetilde{\chi}_h u_1\|_{L^2} + h^{2 - 2\beta} \|\widetilde{\chi}_h u_1\|_{L^2}\Big)\\
			& \quad \leq  \varepsilon \Big(h^{2\beta} \|\widetilde{\chi}_h u_1\|^2_{L^2} + \|(hD_x)(\widetilde{\chi}_h u_1)\|^2_{L^2}\Big) \\
			& \hspace{1cm} + C \varepsilon^{-1}h^{2\beta} \Big(h^{-4\beta}\|\widetilde{\chi}_h v\|^2_{L^2} + C(h^{4 - 8\beta} + h^{2 - 4\beta})\|\underline{\widetilde{\chi}_{h}}u_1\|_{L^2}^2\Big),
		\end{align*}
		for any $\varepsilon \in (0,1)$. Choosing $\varepsilon > 0$ small enough, we deduce that
		\begin{align*}
			h^{2\beta} \|\widetilde{\chi}_h u_1\|^2_{L^2}  \leq Ch^{2\beta} \left(h^{-4\beta} \|\widetilde{\chi}_h v\|^2_{L^2} + h^{2 - 4\beta} \|\underline{\widetilde{\chi}_h} u_1\|^2_{L^2}\right)
		\end{align*}
		i.e.
		$$\|\widetilde{\chi}_h u_1\|^2_{L^2} \leq Ch^{-4\beta} \|\widetilde{\chi}_h v\|^2_{L^2} +C h^{\varepsilon} \|\underline{\widetilde{\chi}_h} u_1\|^2_{L^2},$$
		where $\varepsilon := 2 - 4\beta > 0$. 
		Since $\underline{\widetilde{\chi}_h}$ still satisfies the assumptions of the lemma, this argument can be repeated as many times as necessary as in the proof of Proposition \ref{prop:pseudolocH} (with $\underline{\widetilde{\chi}_h}$ playing the role of $\widetilde{\chi}_h$ in the next iteration). This gives the first estimate in \eqref{eq:suffice_ellip_far}; the second can be shown similarly. 
	\end{proof}
	\begin{lemma}[Estimate of $D_h(z)$]
		\label{lem:controlD}
		Given $R>0$ and $c > 0$, there exist $C,h_0>0$ such that for all $h\in(0,h_0)$ and for all $z\in D(0,Rh)$,
		\[\mathrm{dist}\big(z,\mathrm{sp}(\mathscr{L}^{[2]}_\ell(h))\big)\geq ch \implies \|D_h(z)\|\leq \frac{Ch^{\frac15}}{\mathrm{dist}\big(z,\mathrm{sp}(\mathscr{L}^{[2]}_\ell(h))\big)}\,,\]
		where $D_h(z)$ is defined by \eqref{eq:defDh(z)}.
	\end{lemma}
	\begin{proof}
		Let $z\in D(0,Rh)$ be such that $d:=\textup{dist}\big(z,\mathrm{sp}(\mathscr{L}_\ell^{[2]}(h))\big)\geq ch$; note that $d \leq Ch$ as well since $\textup{dist}(\nu_1(ah),D(0,Rh)) \leq (R +a)h$. Let $\beta \in [0,\frac12)$ be a parameter to be optimized later. By Lemma \ref{lem:elliptic_far}, 
		\begin{equation}
		\label{eq:start_Dh}
		\|D_h(z)\| \leq \|\chi_\ell D_h(z)\| + Ch^{-2\beta} + Ch^N \|(\mathscr{L}_\ell(h) - z)^{-1}\|,
		\end{equation}
		where $\chi_\ell(x)=\chi_0(h^{-\beta}(x-x_\ell))$ and $\chi_0\in C^\infty_c(\R)$ equals $1$ on $[-1,1]$ and satisfies $\supp(\chi_0) \subset [-2,2]$. To estimate $\|\chi_\ell D_h(z)\|$, we write
		\begin{align}
		\nonumber
			\chi_h D_h(z) 
			&= \chi_\ell (\mathscr{L}_\ell^{[2]}-z)^{-1} \underline{\chi_\ell} (V^{[2]}_\ell - V) (\mathscr{L}_\ell - z)^{-1} \\
			& \qquad + \chi_\ell (\mathscr{L}_\ell^{[2]}-z)^{-1}(1 - \underline{\chi_\ell}) (V^{[2]}_\ell - V) (\mathscr{L}_\ell - z)^{-1},
			\label{eq:chih_Dh_2terms}
		\end{align}
		where $\underline{\chi_\ell} := \underline{\chi_0}(h^{-\beta}(x-x_\ell))$ for some $\underline{\chi_0} \in C^\infty(\R)$ chosen such that $\chi_0 \prec \underline{\chi_0}$. Applying the resolvent estimate for $(\mathscr{L}^{[2]}_\ell - z)^{-1}$ (Proposition \ref{prop:resboundH}) in the first term of the right-hand side of \eqref{eq:chih_Dh_2terms}, and the pseudo-locality of $\mathscr{L}_\ell^{[2]}(h)$ (Proposition \ref{prop:pseudolocH} with $\tau = 1$ -- this is legitimate since $1 < \frac32 - \beta$) in the second term, we obtain
		\begin{align*}
			\|\chi_\ell D_h(z)\| 
			&\leq C \left( d^{-1}\|\underline{\chi_\ell} (V - V_\ell^{[2]})\|_{\infty} + h^N\right) \|(\mathscr{L}_\ell - z)^{-1}\| \\
			&\leq Cd^{-1} h^{3\beta}\|V^{(3)}\|_\infty\|(\mathscr{L}_\ell - z)^{-1}\|
		\end{align*}
		Inserting this bound in \eqref{eq:start_Dh}, recalling that $ch \leq d \leq Ch$, and choosing $N$ large enough,
		\begin{align*}
		\|D_h(z)\| &\leq Ch^{3\beta-1}\|(\mathscr{L}_\ell(h)- z)^{-1}\| + h^{-2\beta}\\
		& \leq C h^{3\beta-1} \big(\|(\mathscr{L}^{[2]}_\ell(h) - z)^{-1}\| + \|D_h(z)\|\big) + h^{-2\beta}.
		\end{align*}
		Provided that $\beta$ is such that $3 \beta - 1 > 0$, we deduce that for $h$ small enough,
		$$\|D_h(z)\| \leq C h^{3\beta - 2} + Ch^{-2\beta} \leq C d^{-1} \left(h^{3\beta - 1} + h^{1 - 2\beta}\right).$$
		Optimizing in $\beta$ leads to $\beta := \frac25$ (which indeed lies in $[0,\frac12)$ and satisfies $3\beta - 1 > 0$), giving
		$$\|D_h(z)\| \leq C d^{-1} h^{\frac15}. \qedhere$$
	\end{proof}

%

\subsubsection{Proof of Proposition \ref{prop:fourre_tout}}

\label{sec:proof_2_props}
We first give two immediate consequences of Lemma \ref{lem:controlD}.
\begin{corollary}
\label{cor:almostSpectralTheoremLl}
Given $R > 0$, $c > 0$, there exists $C,h_0 > 0$ such that for all $h \in (0,h_0)$ and for all $z \in D(0,Rh)$, 
	\begin{equation*}
		\mathrm{dist}\big(z,\mathrm{sp}(\mathscr{L}^{[2]}_\ell(h))\big)\geq 	ch \implies \|(\mathscr{L}_\ell(h) - z)^{-1}\| \leq \frac{C}{\mathrm{dist}\big(z,\mathrm{sp}(\mathscr{L}^{[2]}_\ell(h))\big)}.
	\end{equation*}
\end{corollary}
\begin{proof}
	Indeed, by Lemma \ref{lem:controlD} and Proposition \ref{prop:resboundH}, 
	$$\|(\mathscr{L}_\ell - z)^{-1}\| \leq  \|(\mathscr{L}^{[2]}_\ell - z)^{-1}\| + \|D_h(z)\|\leq C\frac{1 + h^{\frac15}}{\textup{dist}(z,\textup{sp}(\mathscr{L}_\ell^{[2]}))} \leq \frac{C'}{\textup{dist}(z,\textup{sp}(\mathscr{L}_\ell^{[2]}))}. \qedhere$$
\end{proof}

\begin{corollary}
	\label{cor:limPi}
	$\lim_{h \to 0} \|\Pi_{\ell,n}(h) - \Pi_{\ell,n}^{[2]}(h)\|_{L^2} = 0.$
\end{corollary}
\begin{proof}
	By Lemma \ref{lem:controlD}, we have for $\varepsilon$ small enough (in view Proposition \ref{prop:locspec}),
	\begin{align*}
	\|\Pi_{\ell,n}(h) - \Pi_{\ell,n}^{[2]}(h)\| 
	&\leq \frac{1}{2\pi} \left|\int_{\mathscr{C}(\nu_n(ah),\varepsilon h)} D_h(z)\,dz\right|\\
	& \leq \frac{1}{2\pi} \frac{C h^{\frac15}}{ \displaystyle\inf_{z \in \mathscr{C}(\nu_n(ah),\varepsilon h)} |z - \nu_{\rev{n}}(ah)|}\cdot 2\pi \varepsilon h =  o_{h \to 0}(1). \qedhere
	\end{align*}
\end{proof}	

\begin{proof}[Proof of Proposition \ref{prop:fourre_tout}]\
\begin{itemize}
\item[(i)] For $h$ small enough, $\|\Pi_{\ell,n}(h) - \Pi_{\ell,n}^{[2]}\| < 1$ by Proposition \ref{cor:limPi}, and thus $\textup{rank}(\Pi_{\ell,n}(h)) = \textup{rank}(\Pi_{\ell,n}^{[2]}(h)) = 1$ by \cite[Lemma 1.5.5]{davies2007linear}. 
\item[(ii)] This follows from Proposition \ref{cor:limPi} by noticing that, thanks to Proposition \ref{prop:resboundH}, 
$$\|\Pi_{\ell,n}^{[2]}(h)\| \leq \frac{1}{2\pi} \frac{2\pi \varepsilon h}{c \varepsilon h} \leq \frac{1}{c}.$$
\item[(iii)] The first statements follow from (ii). Moreover, 
We have $\textup{Ran}(\Pi_{\ell,n}) \supset \textup{Ker}\left(\mathscr{L}_\ell(h) - \mu_{\ell,n}(h)\right)$, and the equality follows since $$1 \leq \dim \big( \textup{Ker}\left(\mathscr{L}_\ell(h) - \mu_{\ell,n}(h)\Id\right)\big) \leq \dim\big(\textup{Ran}(\Pi_{\ell,n}(h))\big) = 1.$$
Next, we have $\textup{Ker}({\Pi}_\ell(h)) = \textup{Ran}(\Pi_\ell(h)^*)^\perp$ and we observe that $\Pi_{\ell,n}(h)^*$ is equal to the Riesz projector of $\mathscr{L}_\ell(h)^*$ associated to the eigenvalue $\overline{\mu_{\ell,n}(h)}$. Thus, if $\psi_\ell \in \textup{Ran}(\Pi_{\ell,n}(h))$, the image of $\Pi_\ell(h)^*$ is spanned by $\overline{\psi_\ell}$, since $$(\mathscr{L}_\ell^* - \overline{\mu_{\ell,n}(h)}) \overline{\psi_\ell} = \overline{(\mathscr{L}_\ell(h) - \mu_{\ell,n}(h))\psi_\ell} = 0.$$
Hence, $\textup{Ker}(\Pi_\ell(h)) = \textup{Span}(\{\overline{\psi_\ell}\})^\perp$ and in particular, there exists $c \neq 0$ such that
$$\Pi_{\ell,n}(h) = c \langle \cdot,\overline{\psi_\ell}\rangle \psi_\ell.$$ 
Finally, one has $c\langle \psi_\ell,\overline{\psi}_\ell\rangle = 1$ since $\Pi_{\ell,n}(h) \psi_\ell = \psi_\ell$. 
\item[(iv)] By Proposition \ref{cor:almostSpectralTheoremLl}, $\sup_{z \in \mathscr{C}(\nu_n(ah),2\varepsilon h)} \|(\mathscr{L}_\ell(h) -z)^{-1}\| \leq Ch^{-1}$. The result thus follows from Proposition \ref{prop:general_spectral_result} with $\Omega := D(\nu_n(ah),2\varepsilon h)$ and $K := \overline{D(\nu_n(ah),\varepsilon h)}$. 
\item[(v)] Given $u \in H^2(\R)$,
we write
\begin{align*}
(\mu_{\ell,n}(h) - z)u &= (\Pi_{\ell,n}(h) + \Id - \Pi_{\ell,n}(h))(\mu_{\ell,n}(h) - z)  u\\
&= \Pi_{\ell,n}(h)(\mathscr{L}_\ell(h) - z)  u + (\mu_{\ell,n}(h) - z)(\Id - \Pi_{\ell,n}(h)) u
\end{align*}
where we have used that $\Pi_{\ell,n}(h) u$ is (equal to zero or) an eigenvector of $\mathscr{L}_\ell(h)$ by (iii), and that $\Pi_{\ell,n}(h)$ commutes with $\mathscr{L}_\ell(h)$. Thus, for $z \in D(\nu_n(ah),\varepsilon h)$,
we deduce by (i) and (iv) that
\begin{align*}
|\mu_{\ell,n}(h) - z|\|u\| &\leq \|\Pi_{\ell,n}(h)\|\|(\mathscr{L}_\ell(h)- z) u\| +  |\mu_{\ell,n}(h) - z| \|(\Id - \Pi_{\ell,n}(h))u\|\\
& \leq C \big(1 + h^{-1}|\mu_{\ell,n}(h) - z|\big) \|(\mathscr{L}_\ell(h) - z)\rev{u}\| \\
& \leq C\|(\mathscr{L}_{\ell}(h) - z)u\|.\qedhere
\end{align*}
\end{itemize}
\end{proof}

	\section{WKB approximation for the simple-well operator}
		
		\label{sec:proofWKB}

		The goal of this section is to give a WKB-type approximation of the low-lying eigenvalues and eigenfunctions of the simple-well operator. This is a classical tool in the analysis of semiclassical tunneling in selfadjoint settings, \cite{HS84, Robert,Hel88, DiSj99, DR25} and the references therein. These approximations take the form 
		$$(\mathscr{L}_\ell(h) - \mu(h)) (e^{-\varphi_\ell(x)/h} a(x,h)) = 0$$
		for some phase function $\varphi_\ell$ and 
		$$a(h)\sim a_0(x) + h a_1(x) + h^2 a_2(x) \ldots\,, \quad \mu(h) \sim \mu_0 + h \mu_1 + h^2 \mu_2 + \ldots.$$
		Cancelling the leading order term, one obtains the eikonal equation $(\varphi_\ell')^2  = e^{i\alpha} V_\ell$, which is satisfied by 
		\begin{equation}
		\label{eq:defPhil}
		\varphi_\ell(x) := e^{i\alpha/2} \abs{\int_{x_\ell}^x \sqrt{V_{\ell}(s)}\,ds}.
		\end{equation}
		Completing the construction, one obtains the following standard result (proven in Appendix~\ref{sec:WKBconstruction}).
		\begin{proposition}[Formal WKB expansion]
			\label{prop:WKBansatz}
			Let $n \in \mathbb{N}^*$. For every $J \in \N^*$, there exists 
			\begin{equation}
				\label{eq:defWKBamplitude}
				\mu_n^{\rm wkb}(h) = \sum_{j = 1}^J \mu_{n,j}h^j\,, \quad 
				a^{\rm wkb}_{n}(x;h) = \sum_{j = 0}^{J} a_{n,j}(x)h^j
			\end{equation}
			with $\mu_{n,j} \in \C$ and $a_{n,j} \in C^\infty(\R)$ and 
			\begin{equation}
				\label{eq:def_mun1}
				\mu_{n,1} := (2n-1) e^{i\alpha/2} 	\sqrt{\frac{V''(x_\ell)}{2}}\,,
			\end{equation}
			\begin{equation}
				\label{eq:def_an0}
				a_{n,0} := 	\left(\frac{\varphi'_\ell(x)}{\varphi_\ell''(x_\ell)}\right)^{n-1}\exp\left(-(2n-1)\int_{x_\ell}^x \frac{\varphi_\ell''(s) - \varphi_\ell''(x_\ell)}{2\varphi_\ell'(s)}\,ds\right)
			\end{equation}
			such that, letting 
			\begin{equation}
			\label{eq:defWKBpsi}
			\psi^{\rm wkb}_n(x;h) := e^{-\frac{\varphi_\ell(x)}{h}} a^{\rm wkb}_n(x;h),
			\end{equation}
			the following holds. 
			For any compact interval $K \subset \R$, there exists $C$, $h_0 > 0$ such that
			$$\big\|e^{\varphi_\ell/h}\big(\mathscr{L}_\ell(h)- \mu_{n}^{\rm wkb}(h)\big) \psi_n^{\rm wkb}(\cdot;h)\big\|_{L^\infty(K)} \leq C h^{J+1}\,,\qquad h \in (0,h_0).$$
			Moreover, there exists $c_n > 0$ such that for any interval $I \subset \R$ containing $x_\ell$ in its interior, there exists $h_0 >0$ such that 
			\begin{equation}
			\label{eq:lowerBoundWKB}
			\big\|\psi_n^{\rm wkb}(\cdot;h)\|_{L^2(I)} \geq c_n h^{\frac{n}{2} - \frac14}, \qquad \textup{for all } h \in (0,h_0),
			\end{equation}
			Finally, one has the estimates
			\begin{equation}
			\label{eq:norm_wkb_1}
			\|\psi_1^{\rm wkb}\| \sim h^{1/4}\left(\frac{\pi}{\cos(\alpha/2)} \sqrt{\frac{2}{V''(x_\ell)}}\right)^{1/4}
			\end{equation}
			\begin{equation}
			\label{eq:psipsibar_wkb}
			\frac{\|\psi_1^{\rm wkb}\|^2}{\langle \psi_1^{\rm wkb},\overline{\psi_1^{\rm wkb}}\rangle} \sim \frac{e^{i\alpha/4}}{\sqrt{\cos (\alpha/2)}}\,.
			\end{equation}
		\end{proposition}
		\begin{definition}[WKB quasimode]
			\label{def:WKBansatz}
			If $a_n^{\rm wkb}(x;h)$, $\mu^{\rm wkb}_n(h)$ and $\psi_n^{\rm wkb}(x;h)$ satisfy the properties of Proposition \ref{prop:WKBansatz}, we say that $\psi_n^{\rm wkb}$ is a {\em WKB quasimode with parameter $(n,J)$}, $a_n(x;h)$ is the {\em WKB amplitude} of $\psi_n^{\rm wkb}$, and $\mu_n^{\rm wkb}(h)$ is the associated {\em WKB approximate eigenvalue}. 
		\end{definition}
		The main result of this section is that, as in the selfadjoint setting, the WKB approximations describe with an \rev{exponential} accuracy the spectrum of $\mathscr{L}_\ell(h)$. More precisely:
		\begin{proposition}[Optimality of the WKB approximations]
			\label{prop:optimalWKB}
			\label{prop:WKB_estim}
			For all $R>0$, there exists $h_0>0$ such that for all $h\in(0,h_0)$, the spectrum of $\mathscr{L}_\ell$ in $D(0,Rh)$ consists of $N$ algebraically simple eigenvalues $\mu_{1}(h),\ldots,\mu_{N}(h)$ satisfying
			\[\mu_{n}(h)=\displaystyle(2n-1)e^{i\frac{\alpha}{2}}h\sqrt{\frac{V''(x_\ell)}{2}}+\mathscr{O}(h^{2})\,, \quad 1 \leq n \leq N\,,\]
			where $N = \lfloor \frac{R+a}{2a}\rfloor$ and $a = \sqrt{\frac{V''(x_\ell)}{2}}$. 
			Moreover, for any $k > 0$, there exists $J$ large enough and $C, h_0 > 0$ such that if $\mu_{n}^{\rm wkb}(h)$ and $\psi_{n}^{\rm wkb}(\cdot;h)$ are WKB Ansätze with parameter $(n,J)$, then for all $h \in (0,h_0)$, 
			\begin{equation}\label{eq.L2approx}
			|\mu_n(h) - \mu^{\rm wkb}_n(h)| \leq C h^k\,, \quad \|\psi_{n}(\cdot;h)- C_n(h)\psi_{n}^{\rm wkb}(\cdot;h)\|_{L^2(\R)} \leq Ch^{k},
			\end{equation}
			for some normalized eigenfunction $\psi_{n}(\cdot;h)$ of $\mathscr{L}_\ell(h)$ and associated to $\mu_{n}(h)$, and some positive constant $C_n(h)$ satisfying 
			$$C_n(h) \sim \frac{1}{\|\psi_n^{\rm wkb}(\cdot;h)\|}.$$	
			Furthermore, for any compact interval $K \subset \R$, there exists $C,h_0 > 0$ such that  
			\begin{equation}\label{eq.approxWKB}
				\Big\|e^{\frac{\textup{Re}(\varphi_\ell)}{h}}(hD_x)^m\big(\psi_n(h)- C_n(h)\psi_{n}^{\rm wkb}(\cdot;h)\big)\Big\|_{L^\infty(K)} \leq C h^k\,, \quad m = 0,1 \,, \quad h \in (0,h_0)\,.\end{equation}
		\end{proposition}
	
		We prove Proposition \ref{prop:WKB_estim} in the case where $n = 1$, since the general case does not present any more difficulties. The result for $n = 1$ follows immediately from Proposition \ref{cor:WKB_approx_n=1} below. The main ingredient for the proof is the construction of a suitable set of subsolutions (in the sense of Definition \ref{def:M_solution}), which we present now.

	\subsection{A particular family of subsolutions}
	
	The presentation of this paragraph is drawn from \cite[Lemma 4.3 \& Corollary 4.5]{DR25}, see also the older original reference \cite{HS84} and the lecture notes \cite[Theorem 4.4.4]{Hel88} and \cite[Prop. A.2]{DiSj99}.
	
	Let $\chi \in C^\infty_c(\R)$ with $\supp(\chi) \subset [-1,1]$ and $\chi \equiv 1$ on $[-\frac{1}{2},\frac12]$, and for $L > 0$, denote 
	$$\chi_{L,h}(x) := \chi\left(h^{-1/2}(x-x_\ell)/L\right).$$ 
	\begin{lemma}[The functions $\Phi_h$]
	\label{lem:optimal_weight}
	Given $A, M > 0$, there exists $L > 0$, $h_0 > 0$ and, a family $(\Phi_h)_{h \in (0,h_0)}$ of smooth, bounded, real-valued functions on $\R$ such that for all $h \in (0,h_0)$:
		\begin{itemize}
		\item[(i)] $\Phi_h$ is an $(Mh)$-subsolution associated to $\chi_{L,h}$ for the potential $V_\ell$.
		\item[(ii)] There exists $c > 0$ such that 
		$$\Phi_h(x) - \textup{Re}(\varphi_\ell(x)) \leq -c\,, \qquad x \in \R \setminus [-2A,2A].$$
		\item[(iii)] There exists $B > 0$ such that 
		$$\Phi_h(x) - \textup{Re}(\varphi_\ell(x)) \geq B h \log(h)\,, \qquad x \in [-A,A].$$
		\end{itemize}
	\end{lemma}
	\begin{proof}
	One can first construct a function $\widetilde{\Phi}_\ell \in C^\infty(\R)$ which is equal to $\Re(\varphi_\ell)$ on $[-A,A]$, constant on $\R \setminus [-2A,2A]$, such that 
	$\widetilde{\Phi}_\ell - \Re(\varphi_\ell)  \leq -c$ on $\R \setminus [-2A,2A]$ for some $c > 0$, and such that $\widetilde{\Phi}_\ell'(x)^2 \leq \Re(\varphi_\ell'(x))^2$ on $\R$. Moreover, we can take $\widetilde{\Phi}_\ell$ such that $\widetilde{\Phi}_\ell'(x)(x-x_\ell) \geq 0$ for all $x \in \R$ (see \cite[Lemma 4.3]{DR25}). We then look for a function $\Phi_h$ saturating the inequality
	$$\Phi_h'(x)^2 - \widetilde{\Phi}_\ell'(x)^2 \leq -Mh (1 - \chi_{L,h}^2)$$
	(which is the requirement for being an $(Mh)$-subsolution up to the irrelevant behavior far from $x_\ell$).  Seeking $\Phi_h'$ under the form $(1 - \varepsilon(x,h))\widetilde{\Phi}'_\ell$ with $\varepsilon(x,h) \approx 0$ leads to
	$$(1 - \varepsilon(x,h))^2 - 1 = \frac{-Mh (1 - \chi_{L,h}(x)^2)}{(\widetilde{\Phi}'_\ell(x))^2} \implies \varepsilon(x,h) \approx\frac{Mh (1 - \chi_{L,h}(x)^2)}{2(\widetilde{\Phi}'_\ell(x))^2}$$
	to first order in $\varepsilon$. This and the next calculations motivate setting 
	$$ \varepsilon(x,h) := \frac{Mh(1 - \chi_{L,h}(x)^2)}{ \Re(\varphi_\ell'(x))^2} \quad \textup{and} \quad \Phi'_h := (1 - \varepsilon(x,h)) \widetilde{\Phi}_\ell'(x)$$
	(it is convenient to have $\Re(\varphi'_\ell)^2$ in the denominator of $\varepsilon(x,h)$ as can be seen below), i.e.,
	\begin{equation}
	\label{eq:defPhih}
	\Phi_h(x) := \widetilde{\Phi}_\ell(x) - Mh \int_{\rev{x_\ell}}^{\rev{x}} \frac{\widetilde{\Phi}'_\ell(s)}{\Re(\varphi_\ell'(s))^2}(1 - \chi_{L,h}^2(s)) \,ds\,.
	\end{equation}

	Using the properties of $\varphi_\ell$, we can choose $L$ large enough and $h_0$ small enough to ensure that $\varepsilon(x,h) \leq \frac12$ for all $x \in \R$ and $h \in (0,h_0)$. Using this and the fact that $\widetilde{\Phi}_\ell'(x)^2 \leq \Re(\varphi_\ell(x))^2$, we obtain
	\begin{align*}
	\Phi'_h(x)^2 - \Re(\varphi_\ell'(x))^2 
	&\leq \Big(\big(1 - \varepsilon(x,h)\big)^2 - 1\Big) \Re(\varphi_\ell'(x))^2\\
	& =  -\left(1 - \frac{\varepsilon(x,h)}{2}\right)2Mh(1 - \chi_{L,h}^2) \\
	& \leq - Mh(1 - \chi_{L,h}^2).
	\end{align*}
	Since $\Re(\varphi'_\ell)^2 = \cos(\alpha/2)^2 V_\ell$, this establishes the property (i).

	Next, since by definition of $\widetilde{\Phi}'_\ell$, the integral term in the right-hand side of \eqref{eq:defPhih} is positive, the property (ii) follows immediately from the definition of $\widetilde{\Phi}_\ell$. 
		
	Finally, since there exists $c > 0$ such that $\Re(\varphi'_\ell(x)) \sim c(x-x_\ell)$ as $x \to x_\ell$ and $\Re(\varphi'_\ell(x)) > 0$ for $x \neq x_\ell$, one can find $g \in C^\infty(\R)$ with $g(x) \geq 0$ such that for all $x \neq x_\ell$,
	$$\frac{1}{\Re(\varphi'_\ell(x))} = \frac{g(x)}{x - x_\ell}.$$ 
	Hence, for all $x \in [-A,A]$ and $h \in (0,h_0)$, we have (assuming that $x > x_\ell$, the case $x < x_\ell$ being similar), 
	\begin{align*}
	\Phi_h(x) - \Re(\varphi_\ell(x)) = \Phi_h(x) - \widetilde{\Phi}_\ell(x) 
	&= -Mh\int_{x_\ell}^{x} \frac{\widetilde{\Phi}_\ell'(s)}{\Re(\varphi'_\ell(s))^2}(1 - \chi_{L,h}^2)\,ds\\
	&\geq -Mh\int_{x_\ell + L/2\sqrt{h}}^{x} \frac{1}{\Re(\varphi_\ell'(s))}\,ds\\
	& = -Mh\int_{x_\ell+L/2\sqrt{h}}^x \frac{g(s)}{s-x_\ell}\,ds\\
	& \geq -Mh(\max_{[-A,A]} g) \log\left(\frac{2A}{L\sqrt{h}}\right)\\
	& \geq Bh \log h
	\end{align*}
	for some constant $B > 0$. This shows the property (iii). 
	\end{proof}
	
	\subsection{Proof of the WKB approximation}
	
	\begin{proposition}[WKB approximation]
		\label{cor:WKB_approx_n=1}
		For every $k > 0$, there exists an integer $J \geq 0$ such that the following holds. For any compact interval $K\subset \R$ containing $x_\ell$ and any smooth compactly supported function $\chi \in C^\infty_c(\R)$ with $\chi \equiv 1$ on $K$, there exist $C$, $h_0 > 0$ such that
		\begin{equation}
		\label{eq:WKB_estimates_weak}
		|\mu_{\ell,1}(h) - \mu^{\rm wkb}(h)| \leq C h^k\,, \quad \| \psi^{\rm wkb} - \Pi_{\ell,1}(\chi \psi^{\rm wkb})\|_{L^2} \leq Ch^{k}\,, 
		\end{equation}
		\begin{equation}
		\label{eq:WKB_estimates_strong}
		\left\|e^{\frac{\textup{Re}(\varphi_\ell)}{h}} D_x^m\big(\psi^{\rm wkb}-\Pi_{\ell,1}(\chi\psi^{\rm wkb})\big)\right\|_{L^\infty(K)}  \leq C h^k\,,
		\end{equation}
		for any $h \in (0,h_0)$, $m \in \{0,1\}$, and any WKB quasimode $\psi^{\rm wkb}(x;h)$ of parameter $(1,J)$ associated to the WKB eigenvalue $\mu^{\rm wkb}(h)$.
	\end{proposition}
	\begin{proof}
	 	Since $\Re(\varphi_\ell) \geq c > 0$ for $x \notin K$, we have
		\begin{align*}
		\big\|(\mathscr{L}_\ell-\mu^{\mathrm{wkb}})(\chi\psi^{\mathrm{wkb}})\big\| 
		& \leq \big\|e^{\Re(\varphi_\ell)/h} (\mathscr{L}_\ell - \rev{\mu}^{\rm wkb})\psi^{\rm wkb}\big\|_{L^2(K)} + e^{-c/h} \big\|(\mathscr{L}_\ell^{\varphi_\ell} - \mu^{\rm wkb})(\chi a)\big\| 
		\end{align*}
		where $a = a(x;h)$ is the WKB amplitude of $\psi^{\rm \rev{wkb}}$, and where $\mathscr{L}_\ell^{\varphi_\ell} := e^{\varphi_\ell/h} \mathscr{L}_\ell e^{-\varphi_\ell/h}$. Hence, since $\mathscr{L}_\ell^{\varphi_\ell}$ is a differential operator with smooth and bounded coefficients (by \eqref{eq:expression_Lphi}), we deduce from Proposition \ref{prop:WKBansatz} that
		$$\|\big(\mathscr{L}_\ell(h) - \rev{\mu}^{\rm wkb}(h)\big) (\chi \rev{\psi}^{\rm wkb})\| \leq C h^J.$$
		By properties (iv) and (v) of Proposition \ref{prop:fourre_tout}, it follows that
		\begin{equation}
		\label{eq:(I-Pi)wkb}
		\|(\Id - \Pi_{\ell,1})(\chi \rev{\psi}^{\rm wkb})\| \leq C h^{J - 1}
		\end{equation}
		\begin{equation}
		\label{eq:mu_wkb-mu}
		|\mu_{\ell,1}(h)-\mu^{\mathrm{wkb}}| \leq C h^{J-\frac14}.
		\end{equation}
	 	since $\|\rev{\psi}^{\rm wkb}\| \geq ch^{\frac14}$ (by \eqref{eq:lowerBoundWKB} in Proposition \ref{prop:WKBansatz}).
		Letting $J > k+1$, this gives the estimates in \eqref{eq:WKB_estimates_weak} using the exponential decay to approximate $\chi \rev{\psi^{\rm wkb}}$ by $\rev{\psi^{\rm wkb}}$ in \eqref{eq:(I-Pi)wkb} up to a neglectable term.
		
		Next, let $A > 0$ be such that $K \subset [-A,A]$, let $M > 2R$ and let $(\Phi_h)_{h \in (0,h_0)}$ be as in Lemma \ref{lem:optimal_weight}. Let $\chi \equiv 1$ on $[-2A,2A]$ and put 
		$$u_h := e^{\Phi_h/h}\big(\chi\psi^{\rm wkb}- \Pi_{\ell,1}(\chi \psi^{\rm wkb})\big) \in H^2(\R).$$ 
		Using the property (ii) of $\Phi_h$ from Lemma \ref{lem:optimal_weight}, the bound \eqref{eq:mu_wkb-mu} on $|\mu_{\ell,1}(h) - \mu^{\rm wkb}|$ above, and the bound on $\Pi_{\ell,1}(h)$ (property (ii) of Proposition \ref{prop:fourre_tout}), 
		\begin{align*}
		&\big\|\big(\mathscr{L}_\ell^{\Phi_h} - \mu^{\rm wkb}\big)u_h\big\|\\
		&\quad \leq \big\|(\mathscr{L}_\ell^{\Phi_h} - \mu^{\rm wkb})(e^{\frac{\Phi}{h}}\chi \psi^{\rm wkb})\big\| + |\mu_{\ell,1}(h) - \mu^{\rm wkb}| \cdot \big\|\Pi_{\ell,1}(\chi \psi^{\rm wkb})\big\|\\
		&\quad \leq \big\|e^{\frac{\Phi_h -\varphi_\ell}{h}} e^{\frac{\varphi_\ell}{h}}\big(\mathscr{L}_\ell- \mu^{\rm wkb} \big) \big(\chi \psi^{\rm wkb}\big)\big\| + Ch^{J-\frac14}\\
		&\quad \leq \big\|e^{\frac{\Phi_h - \Re(\varphi_\ell)}{h}}\big\|_{L^\infty([-A,A])}\,\big\|e^{\frac{\varphi_\ell}{h}}\big(\mathscr{L}_\ell- \mu^{\rm wkb} \big) (\chi \psi^{\rm wkb})\big\|_{L^2([-A,A])}\\
		&\quad\qquad  + \big\|e^{\frac{\Phi_h - \Re(\varphi_\ell)}{h}}\big\|_{L^\infty(\R \setminus [-2A,2A])} \big\|(\mathscr{L}_\ell^{\varphi_\ell} - \mu^{\rm wkb})(\chi a)\big\|_{L^2(\supp \chi) \setminus [-A,A]} +Ch^{J-\frac14}\\
		&\quad \leq C\big\|e^{\frac{\varphi_\ell}{h}}\big(\mathscr{L}_\ell- \mu^{\rm wkb} \big) \psi^{\rm wkb}\big\|_{L^\infty([-A,A])} + Ce^{-c/h} \big\|(\mathscr{L}_\ell^{\varphi_\ell} - \mu^{\rm wkb})(\chi a)\big\|_{L^\infty(\supp \chi)} + Ch^{J -\frac14}.
		\end{align*}
		Recalling again from \eqref{eq:expression_Lphi} that $\mathscr{L}_\ell^{\varphi_\ell}(h)$ is a smooth differential operator and using Proposition \ref{prop:WKBansatz}, we deduce that
		\begin{equation}
		\label{eq:canicule1}
		\big\|\big(\mathscr{L}_\ell^{\Phi_h}(h) - \rev{\mu}^{\rm wkb}(h)\big)u_h\big\| \leq Ch^{J-\frac14}.
		\end{equation}
		Moreover, using \eqref{eq:(I-Pi)wkb} and Lemma \ref{lem:chih_subsolution}, 
		\begin{equation}
		\label{eq:canicule2}
		\|\chi_{L,h}u_h\| \leq C h^{J-1}
		\end{equation}
		By property (i) of Lemma \ref{lem:optimal_weight}, we may apply Proposition \ref{prop:estimee_elliptique_conj}, and conclude from the  bounds \eqref{eq:canicule1} and \eqref{eq:canicule2} that
		$$\|u_{\rev h}\|_{H^2_h} \leq Ch^{J-2}.$$
		By the Sobolev embedding $H^2(K) \subset W^{1,\infty}(K)$, it follows that
		$$\|u_h\|_{W^{1,\infty}(K)} \leq \|u_h\|_{H^2(K)} \leq h^{-2}\|u_h\|_{H^2_h} \leq Ch^{J-4}.$$
		Therefore, by property (iii) of Lemma \ref{lem:optimal_weight},
		\begin{align*}
		\big\|e^{\Re(\varphi_\ell)/h} (\chi \psi^{\rm wkb} - \Pi_{\ell,1}(\chi \psi^{\rm wkb})\big\|_{L^\infty(K)}&\leq \big\|e^{\frac{\Re(\varphi_\ell)-\Phi_h}{h}}\big\|_{L^\infty(K)} \|e^{\Phi_h/h}(\chi\psi^{\rm wkb} - \Pi_{\ell,1}(\chi\psi^{\rm wkb})\|_{L^\infty(K)} \\
		& \leq e^{-Bh\log(h)}\|u_h\|_{W^{1,\infty}(K)}\\
		& \leq Ch^{J-4-B}.
		\end{align*}
		By choosing $J > \rev{k} + 4 + B$, this shows \eqref{eq:WKB_estimates_strong} for $m = 0$. The case $m = 1$ is obtained similarly by writing 
		$$e^{\Phi_h/h} D_x (\chi \psi^{\rm wkb} - \Pi_{\ell,1}( \chi \psi^{\rm wkb})) = D_x u_h - \frac{\Phi'_h}{h} u_h $$ 
		and using the fact that $|\Phi'_h| \leq |\Re(\varphi'_\ell)|$ is bounded on $\R$. This concludes the proof. 	
	\end{proof}
	
	\section{Spectral gap for the double-well operator}
	\label{sec:double_well1D}
	
	In this section, we return to the analysis of the double-well operator $\mathscr{L}(h)$, and complete the proof of Theorem \ref{thm:gap}. In \S \ref{sec:quasi_ortho} we prove the almost orthogonality of the eigenfunctions and projectors associated to the left and right wells. An accurate description of the spectrum of the $\mathscr{L}(h)$ is then obtained in \S \ref{sec:precise_spectrum} (Proposition \ref{prop:precise_spectrum_L}), and the eigenvalue gap is finally estimated in \S \ref{sec:gap}, where the proof of Theorem \ref{thm:gap} can be found.
	
	Let us fix some notation. For $h$ small enough and $n \in \mathbb{N}^*$, we denote by $\widetilde{\mu}_n(h) := \mu_{\ell,n}(h)$ the unique eigenvalue of $\mathscr{L}_\ell(h)$ lying in $D(\nu_n(ah),h^{3/2})$, where we recall that
	$$\nu_n(ah) = (2n-1)e^{i\alpha/2} ah\,, \quad a := \sqrt{\frac{V''(x_\ell)}{2}}$$
	(see Proposition \ref{prop:WKB_estim}).  We will be mainly interested in the case $n = 1$ -- although the arguments apply equally for $n \geq 1$ -- and for this reason we denote $\widetilde{\mu}(h) := \widetilde{\mu}_1(h)$.

		Recall from Section \ref{s:simple_well0} that the definition of the simple-well operator $\mathscr{L}_\ell(h)$ depends on a choice of a parameter $\eta > 0$ and a function $\Sigma_\ell$ supported on $(x_r - \eta,x_r + \eta)$. As in Corollary \ref{cor:quasimodes} and Figure \ref{fig:chiell}, let $\chi_\ell \in C^\infty(\R)$ be a function chosen such that 
		\begin{equation}
		\label{eq:cond_chil}
		\supp(\chi_\ell) \subset (-\infty,x_r]\,, \quad \chi_{\ell|(-\infty,x_r - \eta]} \equiv 1\,, \quad \textup{and} \quad \chi_\ell \perp \Sigma_\ell.
		\end{equation}
		
		Considering the symmetry 
		\begin{equation}
		\label{eq:defsigma}
		(\sigma f)(x) := f(-x),
		\end{equation}
		we define $\chi_r := \sigma \chi_\ell$, $\mathscr{L}_r(h) := \sigma \mathscr{L}_\ell(h) \sigma$ and $\Pi_{r,n} :=  \sigma \Pi_{\ell,n} \sigma$ where $\Pi_{\ell,n}$ is the Riesz projector defined by \eqref{eq:defPiln}. We will often use the counterpart for the right well operator $\mathscr{L}_r(h)$, of the results shown for $\mathscr{L}_\ell(h)$; in this case, we will refer to the results for $\mathscr{L}_\ell(h)$ without further comment. Finally, recall the complex quantity $S(\alpha)$ defined in \eqref{eq:defSalpha} which will be important in what follows.

\subsection{Quasi-orthogonality}
\label{sec:quasi_ortho}
The key fact that permits the study the spectrum of the double-well operator is that the eigenvectors $\psi_\ell$ and $\psi_r$ of the right and left well (which are exponential quasimodes of $\mathscr{L}(h)$ by Corollary \ref{cor:quasimodes}) are almost orthogonal, in the sense of the next proposition.
\begin{proposition}[Quasi-orthogonality of the simple-well eigenfunctions]
	\label{prop:quasi_ortho}
	For $h$ sufficiently small, let $\psi_\ell(h)$ and $\psi_r(h)$ denote eigenvectors of $\mathscr{L}_\ell(h)$ and $\mathscr{L}_r(h)$ associated to the eigenvalue $\widetilde{\mu}_n(h)$. Then for all $\delta > 0$,
	$$\abs{\langle \psi_\ell(h), \psi_r(h)\rangle} = \mathscr{O}(e^{-(\Re S(\alpha)-\delta)/h}).$$
\end{proposition}
\begin{proof}
	Let $\varepsilon > 0$ and let $\Phi_{\ell,\varepsilon}$ be defined by \eqref{ex:Phieps}. Let $\Phi_{r,\varepsilon}$ be defined analogously. By the Cauchy-Schwarz inequality and the Agmon estimate (Corollary \ref{cor:Agmon1})
	\begin{align*}
	\abs{\langle \psi_\ell(h),\psi_r(h)\rangle} 
	&\leq \|e^{-(\Phi_{\ell,\varepsilon} + \Phi_{r,\varepsilon})}\|_\infty \|e^{\Phi_{\ell,\varepsilon}} \psi_\ell(h)\|  \|e^{\Phi_{r,\varepsilon}} \psi_r(h)\|\\
	&\leq C  \|e^{-(\Phi_{\ell,\varepsilon} + \Phi_{r,\varepsilon})}\|_\infty.
	\end{align*}
	But since $V_\ell \geq V$ and $V_r \geq V$, one has 
	$$\Phi_{\ell,\varepsilon}(x) + \Phi_{r,\varepsilon}(x) \geq (1- \varepsilon)^{1/2}\,\Re S(\alpha)\,, \qquad x \in \R.$$
	The conclusion follows by choosing $\varepsilon$ small enough. 
\end{proof}

\begin{corollary}[Quasi-orthogonality of the simple-well Riesz projectors]
	\label{cor:PilPir}
	Let $R > 0$ and let $n \in \rev{\{1,\ldots,N(R)\}}$. \rev{Then 
	for all $\delta > 0$,} 
	$$\Pi_{\ell,n}(h) \Pi_{r,n}(h) = \mathscr{O}(e^{-(\Re S(\alpha)-\delta)/h}).$$
\end{corollary}
\begin{proof}
This follows from Proposition \ref{prop:quasi_ortho} since the range of $\Pi_{\ell/r,n}$ lies in the eigenspace of $\mathscr{L}_{\ell/r}(h)$ (by \eqref{eq:Pi_vrai_projecteur}), using the bound on $\Pi_{\ell,n}$ (property (ii) of Proposition \ref{prop:fourre_tout}).
\end{proof}

We shall also use the following result:
	\begin{lemma}
	\label{lem:loc_Pil}
	Let $R > 0,\gamma > 0$ and let $n \in \{1,\ldots,N(R)\}$ where $N(R)$ is defined as in Proposition \ref{prop:locspec}.\rev{\footnote{\rev{Il suffirait d'écrire ``Let $n \in \N$ and $\gamma > 0$.''}}}  
	Let $\phi_\ell \in C^\infty_c(\R)$ satisfy $\phi_\ell \equiv 1$ on $(-\infty,x_\ell + \gamma]$ for some $\gamma > 0$. Then, there exists $C > 0$ such that
	$$\big\|(1 - \phi_\ell) \Pi_{\ell,n}(h)\big\| + \big\|\Pi_{\ell,n}(h)(1 - \phi_\ell)\big\|  = \mathscr{O}(e^{-C/h}).$$
	\end{lemma}
	\begin{proof}
	Let $u \in L^2(\R)$. By \eqref{eq:Pi_vrai_projecteur}, $\Pi_{\ell,n} u$ is an eigenvector of $\mathscr{L}_\ell$ associated to the eigenvalue $\mu_{\ell,n}(h)$. Thus, by the exponential decay of eigenfunctions (property (ii) of Corollary \ref{cor:small_far}) and the bound on $\Pi_{\ell,n}(h)$ (Proposition \ref{prop:fourre_tout} (ii))
	$$\|(1 - \phi_\ell) \Pi_{\ell,n}(h)u\| \leq Ce^{-C/h}\|\Pi_{\ell,n}(h)u\| \leq C e^{-C/h} \|u\|.$$
	Therefore, $\|(1 - \phi_\ell)\Pi_{\ell,n}\| \leq C e^{-C/h}$, and the same bound follows for $\|\Pi_{\ell,n}(1 - \phi_\ell)\|$ by duality, since all results shown on $\mathscr{L}_\ell$ apply equally to $\mathscr{L}_\ell^*$ by switching $\alpha$ to $-\alpha$. 
	\end{proof}

	\subsection{Spectrum of the double-well operator}
	\label{sec:precise_spectrum}
	
	We establish the following result. Its proof can be found at the end of this paragraph. 
	\begin{proposition}
		\label{prop:precise_spectrum_L}
		Let $R > 0$ \rev{and} let N(R) be as in Proposition \ref{prop:locspec}.
		Then, for all $\delta > 0$, there exists $C > 0$ and $h_0 > 0$ such that \rev{for all $h \in (0,h_0)$},
		$$\textup{sp}(\mathscr{L}(h)) \cap D(0,Rh) \subset \bigcup_{n = 1}^{N(R)} D\big(\widetilde{\mu}_n(h),Ce^{-(\Re S(\alpha)-\delta)/h}\big)$$
		and each small disk $D\big(\widetilde{\mu}_n(h),Ce^{-(\Re S(\alpha)-\delta)/h}\big)$ contains exactly two eigenvalues of $\mathscr{L}(h)$, count\rev{ed} with multiplicities. 
	\end{proposition}	
	
	First, some rough information on the spectrum of $\mathscr{L}(h)$ can be obtained in the same way as in Proposition \ref{prop:locspec}: for every $R > 0$, there exists $h_0 > 0$ and $C > 0$, such that 
	\begin{equation}
	\label{eq:roughLocSpecLh}
	\textup{sp}(\mathscr{L}(h)) \cap D(0,Rh) \subset \bigcup_{n = 1}^{N(R)} D(\nu_n(ah),Ch^{3/2}), \qquad h \in (0,h_0).
	\end{equation}
	This allows again to define Riesz projectors for $\mathscr{L}(h)$. More precisely, given \rev{$n\in \mathbb{N}$}, \rev{\eqref{eq:roughLocSpecLh} ensures that} there exists $\varepsilon > 0$ and $h_0 > 0$ such that
	$$\Pi_n(h) := \frac{1}{2\pi i} \int_{\mathscr{C}(\nu_n(ah),\varepsilon h)} (z - \mathscr{L}(h))^{-1}\,dz$$
	is well-defined for all $h \in (0,h_0)$. Let us denote $\Pi(h) := \Pi_1(h)$ and $\Pi_{\ell/r}(h) := \Pi_{\ell/r,1}(h)$.

	\subsubsection{Approximation of $(\mathscr{L}(h)-z)^{-1}$}

	Our first step to analyse $\Pi_n(h)$ is to approximate the resolvent $(\mathscr{L}(h)-z)^{-1}$ by
	\begin{equation}
	\label{eq:Rh(z)}
	R_h(z) := (\mathscr{L}_\ell(h) - z)^{-1} \phi_\ell + (\mathscr{L}_r(h) - z)^{-1} \phi_r
	\end{equation}
	for $z \in \mathscr{C}(\nu_n(ah),\varepsilon h)$. Here, \rev{$\phi_\ell,\phi_r \in C^\infty(\R,[0,1])$ are such that $\phi_\ell + \phi_r = 1$ and $\phi_{\ell/r} \perp \Sigma_{\ell/r}$ (where $\Sigma_{\ell/r}$ is the well-sealing perturbation introduced in \S\ref{sec:Agmon}) so that $\phi_{\ell/r} \mathscr{L}(h) = \mathscr{L}_{\ell/r}(h)$}. 
	
%

\begin{lemma}
	\label{lem:approx_inverse}
	Let $R > 0$. There exists $h_0$, $\varepsilon$, and $C > 0$ such that \rev{for all $h \in (0,h_0)$}, if $z \in D(0,Rh)$ satisfies $\textup{dist}(z,\textup{sp}(\mathscr{L}_\ell(h))) \geq \varepsilon h$, then  
	$$\|(\mathscr{L}(h) - z)^{-1}\| \leq C h^{-1} \quad \textup{and}\quad \|(\mathscr{L}(h) - z)^{-1} - R_h(z)\| \leq Ch^{-1/2}\,,\qquad h \in (0,h_0)$$
	\rev{with} $R_h(z)$ defined by \eqref{eq:Rh(z)}.
\end{lemma}
\begin{proof}
 	One has
 	\begin{equation}
 	\label{eq:Rh(L-z)}
 	R_{h}(z) (\mathscr{L}(h)-z) = \Id + A_h(z)
 	\end{equation}
	where $A_h(z)$ is given by
	$$A_h(z) = (\mathscr{L}_\ell(h) - z)^{-1} [\phi_\ell,(hD_x)^2] +  (\mathscr{L}_r(h) - z)^{-1} [\phi_r,(hD_x)^2].$$
	Moreover, $\|(\mathscr{L}_\ell(h) - z)^{-1}\| \leq Ch^{-1}$ for all $z \in \gamma_n$ (Corollary \ref{cor:almostSpectralTheoremLl}), and thus, by the same proof as in Proposition \ref{prop:resboundDx}, $\|(\mathscr{L}_\ell(h) - z)^{-1}(hD_x)\| \leq Ch^{-1/2}$. Therefore, there exists $C > 0$ such that for all $h \in (0,h_0)$ and $z \in \gamma_n$, 
	\begin{equation}
		\label{eq:normA}
		\|A_h(z)\| \leq C\sqrt{h}.
	\end{equation}
 	In particular, $\Id + A_h(z)$ is invertible for $h$ small enough, and
	$$(\mathscr{L}(h)-z)^{-1} = (\Id + A_h(z))^{-1} R_h(z).$$
	\rev{Moreover, again by Corollary \ref{cor:almostSpectralTheoremLl},}
	\begin{align*}
	\rev{\sup_{z \in \gamma}}\|R_h(z)\| &\leq \|(\mathscr{L}_\ell(h)-z)^{-1}\phi_\ell\| +\|(\mathscr{L}_r(h)-z)^{-1}\phi_r\|
	\leq Ch^{-1}.
	\end{align*}
	\rev{Using this estimate and \eqref{eq:normA} in \eqref{eq:Rh(L-z)} gives the result.}
\end{proof}

\subsubsection{Approximation of the Riesz projector $\Pi_n(h)$}

As a first consequence of Lemma \ref{lem:approx_inverse}, we learn that $(\mathscr{L}(h) - z)^{-1}$ is ``tame'' on the integration contour $\mathscr{C}(\nu_n(ah),\varepsilon h)$. In the same way as in Section \ref{sec:3}, this allows to see, via the next corollary, that any quasi-mode must lie exponentially close to some element of $\textup{Ran}(\Pi_n(h))$. In particular, since we have two quasi-modes at hand ($\chi_\ell \psi_\ell$ and $\chi_r \psi_r$, see Corollary \ref{cor:quasimodes}), this will allow to see that the rank of $\Pi(h)$ is at least two.
\begin{corollary}
	\label{cor:resolventI-PiL}
	There exists $C$, $h_0$, and $\varepsilon > 0$ such that the estimate 
	$$\|(\Id - \Pi_n(h))u\| \leq Ch^{-1} \big\|\big(\mathscr{L}(h) - z\big)u\big\|$$
	holds for all $h \in (0,h_0)$, $z \in D(\nu_n(ah),\varepsilon h)$ and $u \in H^2(\R)$.
\end{corollary}
\begin{proof}
	This follows from the resolvent estimate in Lemma \ref{lem:approx_inverse} and Proposition \ref{prop:general_spectral_result}. 
\end{proof}

On the other hand, Lemma \ref{lem:approx_inverse} also implies an approximation of $\Pi_n(h)$ by the sum of $\Pi_{\ell,n}(h)$ and $\Pi_{r,n}(h)$, which will  provide a lower bound on the rank of $\Pi(h)$, and play a significant role in the proof of Theorem \ref{thm:gap}.

\begin{proposition}[Approximation of $\Pi_n(h)$]
	\label{prop:PiPilPir}
	Let $R > 0$ and let $n \in \{1,\ldots,N(R)\}$. Then, for all $\delta > 0$,
	\begin{equation}
	\label{eq:approx_Pi_exp}
	\Pi_n(h) = \Pi_{\ell,n}(h) + \Pi_{r,n}(h) + \mathscr{O}(e^{-(S_\eta - \delta)/h}).
	\end{equation} 
	where $S_\eta$ is defined by \eqref{eq:def_Seta}.
	\end{proposition}
	\begin{proof}
	Let $\widetilde{R}(h) := \Pi_n(h) - \Pi_{\ell,n}(h) - \Pi_{r,n}(h)$. Then by the definitions of $\Pi_n(h)$ and $R_h(z)$, 
	\begin{align*}
	\widetilde{R}(h) 
	&= -\Pi_{\ell,n} (1-\phi_\ell) -\Pi_{r,n}(1 - \phi_r) 
	+\frac{1}{2\pi i} \int_{\gamma} \Big((\mathscr{L}(h) - z)^{-1} - R_h(z)\Big)\,dz,	\end{align*}
	and Lemmas \ref{lem:loc_Pil} and \ref{lem:approx_inverse} 
	immediately show that $\|\widetilde{R}_h\| \leq C \sqrt{h}$. 
	
Let us now improve this bound. To lighten the notation, we remove the reference to $n$. We first claim that 
\begin{equation}
\label{eq:PiPil}
\Pi_\ell(h) = \Pi_\ell(h)\Pi(h) + \mathscr{O}(e^{-(S_\eta - \delta)/h})  = \Pi(h) \Pi_\ell(h)+ \mathscr{O}(e^{-(S_\eta - \delta)/h})
\end{equation}
and similarly for $\Pi_r(h)$. Indeed, since eigenvectors of $\mathscr{L}_\ell$ are quasimodes of $\mathscr{L}$ (see Corollary \ref{cor:quasimodes}), and in view of \eqref{eq:Pi_vrai_projecteur}, Corollary \ref{cor:resolventI-PiL} gives 
$$\Pi(h)\Pi_\ell(h)u = \Pi_\ell(h)u + \mathscr{O}(e^{-(S_\eta-\delta)/h}).$$
By applying the above result with $\alpha$ switched to $-\alpha$, we find 
$$\Pi(h)^*\Pi_\ell^*(h)  = \Pi_\ell^*(h) + \mathscr{O}(e^{-(S_\eta-\delta)/h})$$
and \eqref{eq:PiPil} follows by duality. 

Finally, by \eqref{eq:PiPil} and Corollary \ref{cor:PilPir}, and since $\Pi(h)$, $\Pi_r(h)$ and $\Pi_\ell(h)$ are projections, 
\begin{align*}
\widetilde{R}(h)^2 = \widetilde{R}(h) + \mathscr{O}(e^{-(S_\eta-\delta)/h}) \quad \textup{that is} \quad \widetilde{R}(h) (\Id - \widetilde{R}(h)) = \mathscr{O}(e^{-(S_\eta-\delta)/h}).
\end{align*}
For $h$ small enough, we have $\|\widetilde{R}(h)\| < 1$ by what precedes, and thus 
$$\widetilde{R}(h) = (\Id - \widetilde{R}(h))^{-1}\mathscr{O}(e^{-(S_\eta-\delta)/h}) = (\Id + O\sqrt{h})  \mathscr{O}(e^{-(S_\eta-\delta)/h}).$$
This concludes the proof.
\end{proof}

\begin{corollary}
\label{cor:rank_at_most2}
For $h$ small enough, the rank of $\Pi_n(h)$ is at most $2$. 
\end{corollary}
\begin{proof}
Suppose by contradiction that $\textup{rank}(\Pi_n(h)) \geq 3$. Then, letting $\psi^*_\ell \in \textup{Ran}(\Pi_{\ell,n}^*)$ and $\psi^*_r \in \textup{Ran}(\Pi_{r,n}^*)$, one can find a non-zero element of $\textup{Ran}(\Pi_n(h))$ such that $u \in \{\psi_\ell^*,\psi_r^*\}^\perp$. But since $\textup{Ker}(\Pi_{\ell,{n/r}}(h)) = \textup{Ran}(\Pi^*_{\ell,n/r}(h))^\perp$, this implies that $\Pi_{\ell,n}(h) u = \Pi_{r,n}(h) u = 0$, and thus, by Proposition \ref{prop:PiPilPir}, 
$$\|u\| = \|\Pi_n(h) u\| = \|(\Pi_{n}(h) - \Pi_{\ell,n}(h) - \Pi_{r,n}(h)) u\| \leq Ce^{-(S_\eta - \delta)/h}\|u\|.$$
This gives a contradiction when $h$ is small enough.
\end{proof}

\begin{corollary}
	\label{cor:approxPistar}
Take $n=1$.	Let $\psi_r(h)$ be an normalized eigenvector of $\mathscr{L}_r(h)$ associated with $\widetilde{\mu}(h)$. Then, for all $\delta > 0$, 
	\begin{equation}
		\label{eq:approxPistar}
		\Pi(h)^* (\chi_r \psi_r) = \frac{\overline{\psi_r}}{\langle \overline{\psi_r},\psi_r\rangle} + \mathscr{O}(e^{-(S_\eta - \delta)/h}).
	\end{equation}
\end{corollary}
\begin{proof} 
	\rev{Using Proposition \ref{prop:PiPilPir} to express $\Pi$ as the sum of $\Pi_\ell$ and $\Pi_r$}, \rev{and then} inserting the expression for $\Pi_\ell(h)$ and $\Pi_r(h)$ provided by Proposition \ref{prop:fourre_tout} (iii), \rev{using the localisation estimate \eqref{eq:loc vecteur propre}} \rev{(to replace $\chi_\ell \psi_\ell$ by $\psi_\ell$)} \rev{and} Proposition \ref{prop:quasi_ortho} \rev{(quasi-orthogonality of $\psi_\ell$ and $\psi_r$)}, \rev{we obtain} 
	\begin{align*}
		\Pi(h)^* (\chi_r \psi_r) &=  \frac{  \langle \chi_r \psi_r,\psi_\ell \rangle}{\langle \overline{\psi_\ell},\psi_\ell\rangle } \overline{\psi_\ell} + \frac{  \langle \chi_r \psi_r,\psi_r \rangle}{\langle \overline{\psi_r},\psi_r\rangle } \overline{\psi_r} + \mathscr{O}(e^{-(S_\eta - \delta)/h})\\
		& = \frac{ \overline{\psi_r}}{\langle \overline{\psi_r},\psi_r\rangle} + \mathscr{O}(e^{-(S_\eta - \delta)/h}).
	\end{align*}
	where we used that $\frac{1}{\abs{\langle \overline{\psi_\ell},\psi_\ell\rangle}} = \|\Pi_\ell(h)^*\|$ (since $\psi_\ell$ and $\overline{\psi_\ell}$ are normalized) and $\|\Pi_\ell(h)^*\| = \|\Pi_\ell(h)\| \leq C$ (Proposition \ref{prop:fourre_tout} (ii)). 
\end{proof}

To proceed, we consider a normalized \rev{eigenfunction} $\rev{\psi}_\ell(h)$ of $\mathscr{L}_\ell(h)$ associated to $\widetilde{\mu}_n(h)$ for each $h$ and let $\chi_r(h) := \sigma \chi_\ell(h)$ (where $\sigma$ is the symmetry defined in \eqref{eq:defsigma}). Let
\begin{equation}
\label{eq:def_f1f2}
f_1(h) := \Pi_n(h) (\chi_\ell \psi_\ell)\,, \quad f_2(h) := \Pi_n(h) (\chi_r \psi_r)
\end{equation}
and let $G(h) \in \mathbb{C}^{2\times2}$ be the Gram matrix
$$(G(h))_{i,j} := \big\langle f_i(h),f_j(h)\big\rangle.$$	
	
\begin{corollary} 
\label{cor:Gram}
For all $\delta > 0$,
\begin{equation}
\label{eq:f_vs_chipsi}
f_1(h) = \chi_\ell \psi_\ell(h) + \mathscr{O}(e^{-(S_\eta - \delta)/h})\,, \quad f_2(h) = \chi_r \psi_r(h) + \mathscr{O}(e^{-(S_\eta - \delta)/h})\,,
\end{equation}
and
$$G(h) = \begin{pmatrix}
1 &0\\
0 & 1
\end{pmatrix} + \mathscr{O}(e^{-(S_\eta-\delta)/h}).$$
\end{corollary}
\begin{proof}
The estimates in \eqref{eq:f_vs_chipsi} follow from Corollary \ref{cor:resolventI-PiL} applied with $z = \widetilde{\mu}(h)$ and Corollary \ref{cor:quasimodes}. In turn, we deduce that
$$G(h) = \begin{pmatrix}
\|\chi_\ell \psi_\ell\|^2 & \langle \chi_\ell \psi_\ell,\chi_r \psi_r\rangle\\
\langle \chi_r \psi_r,\chi_\ell \psi_\ell \rangle &\|\chi_r \psi_r\|^2 \\
\end{pmatrix} + \mathscr{O}(e^{-(S_\eta - \delta)/h}).$$
The claimed estimate on $G(h)$ then follows from the quasi-orthogonality of $\psi_\ell(h)$ and $\psi_r(h)$ (Proposition \ref{prop:quasi_ortho}) and the estimates $\|\psi_\ell(h) - \chi_\ell \psi_\ell(h)\| = \mathscr{O}(e^{-(S_\eta - \delta)/h})$ and  $\|\psi_r(h) - \chi_r \psi_r(h)\| = \mathscr{O}(e^{-(S_\eta - \delta)/h})$ (property (ii) of Corollary \ref{cor:small_far}). 
\end{proof}

\begin{corollary}
	\label{cor:rank2}
	For $h$ small enough, the rank of $\Pi_n(h)$ is exactly $2$.
\end{corollary}
\begin{proof}
	Both $f_1(h)$ and $f_2(h)$ lie in $\textup{Ran}(\Pi_n(h))$. They are linearly independent for $h$ small enough by Corollary \ref{cor:Gram}. Thus, the rank is at least $2$, and at most $2$ by Corollary~\ref{cor:rank_at_most2}.
\end{proof}

\subsubsection{Proof of Proposition \ref{prop:precise_spectrum_L}}

Let $F(h) := \textup{Ran}(\Pi(h))$, and let us consider the restriction
$$L(h) := \mathscr{L}(h)|_{F(h)} : F(h)\to F(h).$$
Since $F(h)$ is of dimension $2$ by Corollary \ref{cor:rank2}, the functions $g_1(h)$ and $g_2(h)$ defined by $(g_1(h),g_2(h))^T := G(h)^{-1/2}(f_1(h),f_2(h))^{T}$ provide an orthonormal basis of $F(h)$. Therefore, the matrix of $L(h)$ in this basis is given by $(\mathcal{M}_g(h))_{ij} = \langle \mathscr{L} g_i,g_j \rangle$. The matrix $\mathcal{M}_g(h)$ can be expressed as
$$\mathcal{M}_{g}(h) = G(h)^{-1/2} \mathcal{M}_f(h) G(h)^{-1/2} =   \mathcal{M}_f(h)\left(1 + \mathscr{O}(e^{-(S_\eta - \delta)/h})\right)$$
(by Corollary \ref{cor:Gram}), where 
$(\mathcal{M}_f(h))_{i,j} := \big\langle \mathscr{L}f_i,f_j\big\rangle$.
By Corollary \ref{cor:quasimodes} and using \eqref{eq:f_vs_chipsi},
$$\rev{\mathcal{M}}_f(h) = \rev{\widetilde{\mu}}(h)I_2 + \mathscr{O}(e^{-(S_\eta - \delta)/h})$$
for all $\delta > 0$, and therefore, 
$$\mathcal{M}_g(h) = \widetilde{\mu}(h) I_2 + \mathscr{O}(e^{-(S_\eta - \delta)/h}).$$
It is well-known that \rev{the} spectrum of $\mathscr{L}(h)$ in $D(\nu_1(ah),\varepsilon h)$ coincides with the spectrum of $L(h)$, which is the same as the spectrum of the matrix $\mathcal{M}_g(h)$. The conclusion follows since the choice of $\eta > 0$ small enough in Section \ref{sec:Agmon} was arbitrary and $S_\eta \to \Re S(\alpha)$ as $\eta \to 0$ (see eq.~\eqref{eq:def_Seta}). \qed

\subsection{Eigenvalue gap}
\label{sec:gap}

By Proposition \ref{prop:precise_spectrum_L}, there are exactly two eigenvalues of $\mathscr{L}(h)$ in $D(\nu_1(ah),\varepsilon h)$; let us denote them in some arbitrary order by $\mu_{1,1}(h)$ and $\mu_{1,2}(h)$, and let
$$\textup{gap}(h) := \mu_{1,1}(h) - \mu_{1,2}(h).$$
Note that this gap can be complex. It can first be expressed, up to an exponentially small error, in terms of the quasimodes $\chi_\ell \psi_\ell$ and $\chi_r \psi_r$.
\begin{proposition}
\label{prop:gap1}
Let $\psi_\ell(h)$ be a normalized eigenvector of $\mathscr{L}_\ell(h)$ associated to $\widetilde{\mu}(h)$, let $\chi_\ell(h)$ satisfy \eqref{eq:cond_chil} and let $\chi_r := \sigma \chi_\ell$, $\psi_r := \sigma \psi_\ell$. Then for all $\delta > 0$, and up to a relabelling of $\mu_{1,1}(h)$ and $\mu_{1,2}(h)$, 
$$\textup{gap}(h) = 2 \Delta(h) + \mathscr{O}(e^{-(2S_\eta-\delta)/h})\,,$$
where 
$$\Delta(h) :=\frac{\big \langle \big(\mathscr{L}(h) - \widetilde{\mu}(h)\big) \chi_\ell \psi_\ell, \overline{\psi_r}\big \rangle}{\langle {\psi_r},\overline{\psi_r}\rangle}\,.$$
\end{proposition}
\begin{proof}
Keeping the notation of the previous paragraph, we have seen that the spectrum of $\mathscr{L}(h)$ in $D(\nu_1(ah),\varepsilon h)$ is the spectrum of the matrix 
$$\mathcal{M}_g(h) := G(h)^{-1/2} \mathcal{M}_f(h) G(h)^{-1/2}.$$
Then, $\textup{gap}(h)$ is equal to the difference between the eigenvalues of $\mathcal{M}_g(h)$. Let us write 
$$G(h) =: I_2 + R_1(h)\,, \quad \mathcal{M}_f(h) =: \widetilde{\mu}(h) I_2 + R_2(h). $$ 
We have also seen that $R_1(h),R_2(h) = \mathscr{O}(e^{-(S_\eta-\delta)/h})$ for all $\delta > 0$. Therefore, $(I_2 + R_1(h))^{-1/2} = I_2 - \frac12 R_1(h) + \mathscr{O}(e^{-2(S_\eta - \delta)/h})$, leading to
\begin{align*}
\mathcal{M}_g(h)  &= \widetilde{\mu}(h)I_2 + R_2(h) - \widetilde{\mu}(h)R_1(h) + \mathscr{O}(e^{-2(S-\delta)/h})\\
& =\widetilde{\mu}(h)I_2 +  \big(\mathcal{M}_f(h) - \widetilde{\mu}(h)I_2\big) - \widetilde{\mu}(h) \big(G(h) - I_2\big) + \mathscr{O}(e^{-C/h})\\
& =\widetilde{\mu}(h)I_2 + \mathcal{M}_f(h) - \widetilde{\mu}(h) G(h) + \mathscr{O}(e^{-2(S_\eta - \delta)/h})\,.
\end{align*}
One can check that $\sigma \mathscr{L}(h) \sigma = \mathscr{L}(h)$ and $\langle \sigma f,g\rangle = \langle f,\sigma g\rangle$ (where $\sigma$ is the symmetry defined in \eqref{eq:defsigma}) and thus, that the matrices $\rev{\mathcal{M}}_g(h)$, $\mathcal{M}_f(h)$ and $G(h)$ belong to $E := \textup{Span}(\{I_2,J_2\})$ where 
$$J_2 = \begin{pmatrix}
0 & 1 \\
1 & 0
\end{pmatrix}.$$
For $M = a I_2 + b J_2$, the eigenvalue gap of $M$ is given, for some choice of labelling, by $2b$. Therefore, if we denote by $\widetilde{\textup{gap}}(h)$ the eigenvalue gap of the matrix $R(h) := \mathcal{M}_f(h) - \widetilde{\mu}(h)G(h)$, we then have (up to a relabelling)
$$\textup{gap}(h) = \widetilde{\textup{gap}}(h)+ \mathscr{O}(e^{-2(S_\eta - \delta)/h}).$$

From the definitions of $\rev{\mathcal{{M}}}_f$ and $G$, the coefficients of $R(h)$ are given by 
$$(R(h))_{ij} = \langle \big(\mathscr{L}(h) - \widetilde{\mu}(h)\big) f_i(h),f_j(h)\rangle,$$
and by what precedes, if we denote $b(h) := (R(h))_{12} = (R(h))_{21}$, then $\widetilde{\textup{gap}}(h) = 2b(h)$ (up to a relabelling). Hence,
\begin{equation}\label{eq.gaph}\textup{gap}(h) =2 b(h) +  \mathscr{O}(e^{-2(S_\eta - \delta)/h})\,.\end{equation}
Now, since $\|(\Id - \Pi)(\chi_r \psi_r)\| = \mathscr{O}(e^{-(S_\eta - \delta)/h})$ and $\|(\mathscr{L} - \widetilde{\mu})\Pi\| = \mathscr{O}(e^{-(S-\delta)/h})$ by Proposition \ref{prop:precise_spectrum_L}, we can remove the projection on $f_{2}(h)$ in $b(h)$ (see \eqref{eq:def_f1f2}) and we get:
\begin{equation}\label{eq.bh}
	\begin{split}
	b(h) &= \big \langle (\mathscr{L}(h) - \widetilde{\mu}(h))\Pi(h)(\chi_\ell\psi_\ell),\chi_r \psi_r\big\rangle + \mathscr{O}(e^{-2(S_\eta - \delta)/h})\\
	&= \big \langle (\mathscr{L}(h) - \widetilde{\mu}(h))(\chi_\ell\psi_\ell),\Pi(h)^*(\chi_r \psi_r)\big\rangle + \mathscr{O}(e^{-2(S_\eta - \delta)/h})\,.
	\end{split}
\end{equation}
Therefore,
\[\big \langle (\mathscr{L}(h) - \widetilde{\mu}(h))(\chi_\ell\psi_\ell),\Pi(h)^*(\chi_r \psi_r)\big\rangle =\Delta(h)+ R\,,\] 
where, by Corollaries \ref{cor:quasimodes} and \ref{cor:approxPistar},
$$\|R\| \leq \|\left(\mathscr{L}(h) - \widetilde{\mu}(h)\right) \chi_\ell\psi_\ell\| \|\Pi(h)^* \psi_r - \langle\overline{\psi_r} ,\psi_r,\rangle^{-1} \overline{\psi_r}\| = \mathscr{O}(e^{-2(S_\eta - \delta)/h})\,.$$
Combining this with \eqref{eq.gaph} and \eqref{eq.bh} concludes the proof.
\end{proof}

We now give an expression of $\Delta(h)$ in terms of the Wronskian
\[W(x,h) := \psi_\ell(x)[(hD_x) \psi_r](x) - \psi_r(x) [(hD_x)\psi_\ell](x)\,,\]
where $\psi_\ell$ and $\psi_r$ are as in Proposition \ref{prop:gap1}.
\begin{lemma}
\label{lem:wronskien}
\[\Delta(h) = \frac{ihW(0;h)}{\langle \psi_r,\overline{\psi_r}\big\rangle}\,.\]
\end{lemma}
\begin{proof}
We first observe that $W(\cdot,h)$ is constant on $[x_\ell + \eta,x_r-\eta]$. Indeed,
\begin{align*}
(hD_x)W(x,h) &= \psi_r (hD_x)^2 \psi_\ell  - \psi_\ell (hD_x)^2\psi_r\\
&= \psi_\ell\psi_r  \left((\widetilde{\mu}(h) - e^{i\alpha}V_\ell) - (\widetilde{\mu}(h) - e^{i\alpha}V_r)\right)\\
& = \psi_\ell\psi_r e^{i\alpha} (\Sigma_r - \Sigma_\ell)\,.
\end{align*}
In particular, $W(\cdot,h)$ is constant on $\supp \chi'_\ell$. 

Next, noticing that 
$$(\mathscr{L}(h) - \widetilde{\mu}(h)) (\chi_\ell \psi_\ell) = (\mathscr{L}_\ell(h) - \widetilde{\mu}(h)) (\chi_\ell \psi_\ell) = [\mathscr{L}_\ell(h),\chi_\ell] \psi_\ell = [(hD_x)^2,\chi_\ell] \psi_\ell,$$
we deduce that
\begin{equation}
\label{eq:intermediate_delta}
\Delta(h) =  \frac{\big \langle [(hD_x)^2, \chi_\ell] \psi_\ell, \overline{\psi_r}\big \rangle}{\langle \psi_r,\overline{\psi_r}\rangle}\,.
\end{equation}
Therefore, 
\begin{align*}
\Big \langle [(hD_x)^2,\chi_\ell] \psi_\ell,\overline{\psi}_r\Big \rangle 
&= \Big\langle \chi_\ell \psi_\ell,(hD_x)^2 \overline{\psi}_r \Big\rangle - \Big\langle \chi_\ell (hD_x)^2 \psi_\ell, \overline{\psi}_r \Big\rangle\\
&= -\langle \chi_\ell,(hD_x)\overline{W}\rangle\\
& = -ih W(0) \langle \chi_\ell',1\rangle\\
& = -ih W(0) \big(\chi(\infty)- \chi(-\infty)\big) = ihW(0).
\end{align*}
Inserting this in \eqref{eq:intermediate_delta} concludes the proof.
\end{proof}

In view of Lemma \ref{lem:wronskien}, it remains to estimate $W(0)$ and $\langle \psi_r,\overline{\psi_r}\rangle$ when $h\to 0$. This can be done thanks to the WKB approximation in Proposition \ref{prop:WKB_estim}.

\begin{lemma}
\label{lem:valeurW(0)}
One can choose  the normalization of $\psi_\ell$ in Proposition \ref{prop:gap1} so that
\begin{equation}
\label{eq:estimW(0)}
W(0)= (w_0 + o(1)) h^{-1/2}e^{-\frac{S(\alpha)}{h}}\,, \quad \textup{and} \quad 
\langle \psi_r,\overline{\psi_r}\rangle = e^{-i\alpha/4}\sqrt{\cos(\alpha/2)}(1 + o(1))\,,
\end{equation}
where 
\[S(\alpha) =e^{i\alpha/2} \int_{x_\ell}^{x_r} \sqrt{V(s)}\,ds=2\varphi_\ell(0)\,,\]
and
\[w_0 = -2ie^{i\alpha/2}\sqrt{V(0)}\left(\frac{\cos(\alpha/2)}{\pi} \sqrt{\frac{V''(x_\ell)}{2}}\right)^{1/2} \exp\left(-2\int_{x_\ell}^0 \frac{(\sqrt{V})'(s) -  \sqrt{V''(x_\ell)/2}}{\sqrt{V(s)}}\,ds\right)\,.\]
\end{lemma}

\begin{proof}
We notice that, since $\psi_r(x) = \psi_\ell(-x)$, 
\begin{equation}
\label{eq:W(0)}
W(0) = 2ih \psi_\ell(0) \psi_\ell'(0)\,.
\end{equation}

Let $k > 0$ be large enough and let $J > 0$ be sufficiently large to apply Proposition \ref{prop:WKB_estim}. Let $\psi^{\rm wkb}_1=a(x,h)e^{-\varphi_\ell/h}$ be a WKB quasimode with parameter $(1,J)$ for $\mathscr{L}_\ell$ (recall Definition \ref{def:WKBansatz}). Let $K \subset \R$ be a compact containing $[x_\ell,x_r]$ and let $\chi \equiv 1$ on $K$. By Proposition \ref{prop:WKB_estim}, we can find a normalized eigenvector  $\psi_\ell$ of $\mathscr{L}_\ell$ associated with $\widetilde{\mu}(h)$ and a constant $C_1(h) > 0$ such that   

$$\|\psi_\ell - C_1(h) \psi^{\rm wkb}_1\|_{L^2} \leq Ch^k\,,$$
\begin{equation}
\label{eq:restate_wkb_approx}
\Big\|e^{\varphi_\ell/h}(hD_x)^m \left[ \psi_\ell -  C_1(h) \psi^{\rm wkb}_1\right]\Big\|_{L^\infty(K)} = \mathscr{O}(h^k)\,, \quad m = 0,1,
\end{equation}
with
\begin{equation}
\label{eq:C1(h)}
C_1(h) \sim \frac{1}{\|\psi_1^{\rm wkb}\|} \sim  h^{-1/4} \left(\frac{\cos(\alpha/2)}{\pi} \sqrt{\frac{V''(x_\ell)}{2}}\right)^{1/4}\,,
\end{equation}
where the last estimate is given by \eqref{eq:norm_wkb_1} in Proposition \ref{prop:WKBansatz}.
The estimate \eqref{eq:restate_wkb_approx} implies 
\begin{align}
\nonumber
\psi_\ell(0) 
&= e^{-\frac{\varphi_\ell(0)}{h}} e^{\frac{\varphi_\ell(0)}{h}} \psi_\ell(0) \\\nonumber
&= e^{-\frac{S(\alpha)}{2h}} C_1(h)  a(0;h) +  e^{-\frac{S(\alpha)}{2h}} \left(e^{\frac{\varphi_\ell(0)}{h}} \left[\psi_\ell(0) - C_1(h)e^{-\frac{\varphi_\ell(0)}{h}}a(0;h)\right]\right) \\
& = e^{-\frac{S(\alpha)}{2h}}C_1(h)  (a_0 + \mathscr{O}(h^k))\,,
\label{eq:psil(0)}
\end{align}
where $a_0 := a_{1,0}(0)$ (see \eqref{eq:def_an0}) is given by
$$a_0 := \exp\left(-\int_{x_\ell}^0 \frac{\varphi_\ell''(s) - \varphi_\ell''(x_\ell)}{2\varphi_\ell'(s)}\,ds\right)\,.$$
Similarly, 
\begin{align}
\nonumber
\psi_\ell'(0) 
&\sim C_1(h) 
\rev{\frac{d}{dx}} \left(a(x;h) e^{-\varphi_\ell/h}\right)\big|_{x=0}\\
& \sim -C_1(h) h^{-1}e^{i\alpha/2}\sqrt{V(0)}e^{-\frac{S(\alpha)}{2h}}a_0.
\label{eq:psil'(0)}
\end{align}
The combination of \eqref{eq:psil(0)}, \eqref{eq:psil'(0)}, \eqref{eq:C1(h)} and \eqref{eq:W(0)} gives the estimate for $W(0)$ in  \eqref{eq:estimW(0)}. 

Finally, thanks to \eqref{eq:restate_wkb_approx}, we have
$$\langle \psi_r,\overline{\psi_r}\rangle \sim \frac{\langle \psi^{\rm wkb},\overline{\psi^{\rm wkb}}\rangle}{\|\psi^{\rm wkb}\|^2}$$
which gives the second estimate in \eqref{eq:W(0)} by using \eqref{eq:psipsibar_wkb}.
\end{proof}

\subsection{End of the proof of Theorem \ref{thm:gap}}
By Lemmas \ref{lem:wronskien} and \ref{lem:valeurW(0)}, we have 
\[\Delta(h) = i\sqrt{h}\Big(\frac{e^{i\alpha/4}}{\sqrt{\cos(\alpha/2)}}w_0  + o(1)\Big)e^{-S(\alpha)/h} = \frac{\sqrt{h}}{2}(\mathsf{A} + o(1))e^{-S(\alpha)/h}\,,\]
where $\mathsf{A}$ given by \eqref{eq:defA}. The estimate \eqref{eq:gap} follows by Proposition \ref{prop:gap1}, noting that $\mathsf{A} \neq 0$. The estimate \eqref{eq:separation} follows from Proposition \ref{prop:precise_spectrum_L}.

\subsection*{Acknowlegments}
This work was conducted within the France 2030 framework programme, Centre Henri Lebesgue ANR-11-LABX-0020-01. N.F. thanks the Région Pays de la Loire for the Connect Talent Project HiFrAn 2022 07750 led by Clotilde Fermanian Kammerer. \rev{The authors are thankful to the anonymous referee for their interesting comments, and to François Moncler for reporting a mistake in a previous version of this paper.}

\appendix

\section{Construction of the WKB quasimodes}
	\label{sec:WKBconstruction}
	
	For the sake of completeness, we recall here the WKB construction. Since we are in the context of differential operators in 1D, we take this opportunity to flesh out all the important details. We closely follow the presentation of \cite[Proposition 2.4]{DR25}.
	
	Let us first record some useful properties of the complex WKB phase:
	\begin{enumerate}
		\item $\varphi_\ell \in C^\infty(\R)$ (using that $V_\ell(x_\ell) = V'_\ell(x_\ell) = 0$ and $V''(x_\ell) \neq 0$),
		\item $\varphi_\ell(x) \sim e^{i\alpha/2} \sqrt{\frac{V''(x_\ell)}{8}}(x-x_\ell)^2$ as $x \to x_\ell$.  
		\item $\varphi'_\ell = e^{i\alpha/2} \sqrt{V_\ell}$ is bounded, as well as all of its derivatives. 
		\item for all $\delta > 0$, $\inf_{|x-x_\ell| \geq \delta} \textup{Re}(\varphi_\ell(x)) > 0$. 
	\end{enumerate}	
	
	Let
	$$\mathscr{L}^{\varphi_\ell}_\ell(h) := e^{\varphi_\ell/h}\mathscr{L}_\ell(h) e^{-\varphi_\ell/h}$$
	where $\varphi_\ell$ is given by \eqref{eq:defPhil}. We shall construct a family of functions $a_n(x;h) \in C^\infty(\R)$ and of coefficients $\mu_n(h) \in \C$ 
	of the form
	$$a^{\rm wkb}_n(x;h) = \sum_{j = 0}^{J-1} a_{n,j}(x) h^j\,, \quad \mu^{\rm wkb}_n(h) = \sum_{j = 1}^{J} \mu_{n,j} h^j,$$
	with $\mu_{n,1}$ and $a_{n,0}$ given by \eqref{eq:def_mun1} and \eqref{eq:def_an0}, and where $a_{n,j} \in C^\infty(\R)$ and $\mu_{n,j} \in \mathbb{C}$. These functions and coefficients will be chosen so that 
	\begin{equation}
	\label{eq:WKB_N}
	(\mathscr{L}_\ell^{\varphi_\ell}(h) - \mu^{\rm wkb}_n(h)) a^{\rm wkb}_n(x,h) = \sum_{k = J}^{2J-1} h^k r_k(x) ,
	\end{equation}
	for some functions $r_k \in C^\infty(\R)$, $k = J,\dots,2J-1$. 
	
	In view of \eqref{eq:expression_Lphi}, let us rewrite the operator in the left-hand side of \eqref{eq:WKB_N} as
	$$\mathscr{L}_\ell^{\varphi_\ell} - \mu_n(h) = \sum_{j = 1}^J h^j (L_j(x,\partial_x)- \mu_{n,j})$$
	where the differential operators $L_j(x,\partial_x)$ are given by
	$$L_j(x,\partial_x) = \begin{cases}
	2\varphi_\ell'(x) \dfrac{d}{dx} + \varphi_\ell''(x)& \textup{for } j = 1\\[0.5em]
	-\dfrac{d^2}{dx^2}  & \textup{for } j = 2\\[0.5em]
	0 & \textup{for } j \geq 3.
	\end{cases}$$
	Applying the operator to the Ansatz $a_n(x;h)$ and sorting by increasing powers of $h$, we find
	$$(\mathscr{L}_\ell^{\varphi_\ell} - \mu^{\rm wkb}_n(h)) a^{\rm wkb}_{n}(x;h) = \sum_{k = 1}^{2N-1} h^k \sum_{j = 1}^{k} (L_j(x,\partial_x) - \mu_{n,j}) a_{n,k - j}(x),$$
	(with the convention that $\mu_{n,j} := 0$ for $j \geq N$ and $a_{n,j} := 0$ for $j \geq N+1$). Thus, for \eqref{eq:WKB_N} to hold, it suffices that 
	\begin{equation}
	\label{eq:WKB_transport}
		\sum_{j = 1}^{k} (L_j(x,\partial_x) -\mu_{n,j}) a_{n,k - j} = 
		0 \quad \textup{for } k = 1,\ldots,N-1\,.
	\end{equation}
	This may be rewritten as 
	\begin{equation}
	\label{eq:WKB_transport_bis}
	\left\{
		\begin{array}{lcl}
		\displaystyle(L_1(x,\partial_x) - \mu_{n,1})a_{n,0} &=& 0 \\[0.5em]
		\displaystyle(L_1(x,\partial_x) - \mu_{n,1})a_{n,1} &=& f_2 - \mu_{n,2} a_{n,0} \\
		& \vdots& \\
		\displaystyle(L_1(x,\partial_x) - \mu_{n,1})a_{n,J-1} &=& f_N - \mu_{n,J} a_{n,0} \\
		\end{array}\right.
	\end{equation}
	where, for all $k \in \{2,\ldots,N\}$, 
	$$f_k := L_k(x,\partial_x)a_{n,0} + \sum_{j = 2}^{k-1} (L_j(x,\partial_x) - \mu_{n,j} )a_{n,k-j}.$$
	
	Noticing that $f_k \in C^\infty(\R)$ only depends on $a_{n,0},\ldots,a_{n,k-1}$ and $\mu_{n,1},\ldots,\mu_{n,k}$, the existence of a set of functions $a_{n,j}$ and coefficients $\mu_{n,j}$ satisfying \eqref{eq:WKB_transport_bis} follows immediately from Lemma \ref{lem:WKB_ordre0} (to see that the first ODE indeed holds with the chosen $a_{n,0}$ and $\mu_{n,1}$) and Lemma \ref{lem:WKB_ordreN} (to solve the remaining ODEs successively) below. Note that Lemma \ref{lem:WKB_ordre0} additionally shows that there are no other solutions to the first ODE than those provided by \eqref{eq:def_mun1} and \eqref{eq:def_an0} (up to a multiplicative constant). 
	
	One can then use Taylor expansions of $a_{n,j}(x)$ at $x = x_\ell$, to write
	$$a_n(x;h) = p(x-x_\ell,h) + (x - x_\ell)^N v(x,h)$$
	where $p(X,h)$ is a polynomial in both variables $X$ and $h$, and $v$ is bounded on $\R \times [0,1]$. Moreover, we have $p(X,0) = a_n(x_\ell + X;0) = U_n(x_\ell + X)$ and thus,
	$$\frac{\partial^{n-1}p}{\partial X^{n-1}}(0,0) = \frac{d^{n-1}}{dx^{n-1}}a_{n,0}(x_\ell) = 1.$$
	Therefore, one can apply Lemma \ref{lem:normWKB} below to obtain the lower bound \eqref{eq:lowerBoundWKB}. The normalization constants \eqref{eq:norm_wkb_1} and \eqref{eq:psipsibar_wkb} are shown in Lemma \ref{lem:constants_n=1}.

	\begin{lemma}[Solutions of the first ODE in \eqref{eq:WKB_transport_bis}]
	\label{lem:WKB_ordre0}
	Let $\lambda \in \C$. Then, the following assertions are equivalent
	\begin{itemize}
	\item[(i)] There exist non-trivial smooth solutions to 
	$$(L_1(x,\partial_x) - \lambda) u = 0 \quad \textup{on } \R.$$
	\item[(ii)] There exists $n \in \mathbb{N}_{\geq 1}$ such that $\lambda  = (2n-1) \varphi_\ell''(x_\ell) = (2n - 1) \sqrt{\frac{V_\ell''(x_\ell)}{2}}e^{i\alpha/2}$. 
	\end{itemize}
	When these conditions are satisfied, the vector space of solutions of $(L_1(x,\partial_x) - \lambda) u = 0$ is spanned by the function $a_{n,0}$ defined by \eqref{eq:def_an0}. 
	\end{lemma}
	\begin{proof}
		Assume that (i) holds and pick $a \in \R \setminus \{x_\ell\}$ such that $u(a) \neq 0$. For example, assume that $a \in (x_\ell,+\infty)$ (the case $a \in (-\infty,x_\ell)$ is similar). Then $u$ is given on $(x_\ell,+\infty)$ by
		$$u(x) = u(a) U_a(x)\,, \quad \textup{where}\quad U_a(x) := \exp\left(-\int_{a}^x \frac{\varphi_\ell''(s) - \lambda}{2\varphi_\ell'(s)}\,ds\right).$$
		Rewriting 
		$$\frac{\varphi''_\ell(s) - \lambda}{2\varphi_\ell'(s)} = \underbrace{\frac{\lambda}{2\varphi''_\ell(x_\ell)} \left(\frac{\varphi''_\ell(s) - \varphi''_\ell(x_\ell)}{\varphi_\ell'(s)}\right)}_{g(s)} + \left(\frac12 - \frac{\lambda}{2\varphi''_\ell(x_\ell)}\right)\frac{\varphi''_\ell(s)}{\varphi'_\ell(s)}$$
		and noticing that $g \in C^\infty(\R)$, we see that 
		$$U_a(x) = \exp\left(-\int_{a}^x g(s)ds\right) \left(\frac{\varphi_\ell'(x)}{\varphi_\ell'(a)}\right)^{\frac{\lambda}{2\varphi_\ell''(x_\ell)}-\frac12}.$$
		For (i) to hold, it is necessary that all derivatives of $U_a$ be bounded near $x_\ell$, and since there exists $C > 0$ such that $\varphi_\ell'(x) \sim C(x - x_\ell)$ near $x_\ell$, this is only possible if $\frac{\lambda}{2\varphi_\ell''(x_\ell)} - \frac12 \in \mathbb{N}$, i.e., if (ii) holds. Reciprocally, if (ii) holds, one can directly verify that $a_{n,0}$ provides a particular solution on $\R$, showing that (i) is satisfied.
	\end{proof}
	\begin{lemma}[Solutions of the remaining ODEs in \eqref{eq:WKB_transport_bis}]
	\label{lem:WKB_ordreN}
	Let $f \in C^\infty(\R)$, let $n \in \mathbb{N} \geq 1$, let
	$\lambda := (2n - 1) \Phi''(x_\ell)$, and let $U_n$ be a non-trivial solution of $(L_1(x,\partial_x) - \lambda) U_n = 0$. Then there exists $\lambda \in \mathbb{C}$ and $u \in C^\infty(\R)$ such that  
	\begin{equation}
	\label{eq:WKB_ordreN}
	(L_1(x,\partial_x) - \lambda) u = f - \lambda U_n.
	\end{equation}
	\end{lemma}
	\begin{proof}
	Let $\lambda \in \mathbb{C}$ be fixed. Seeking $u$ under the form $u(x) = u_0(x)v(x)$, we see that if $u \in C^\infty(\R)$, then $u$ solves \eqref{eq:WKB_ordreN} if and only if
	\begin{equation}
	\label{eq:kind_of_residue}
	\partial_x v(x) = \frac{1}{2\varphi_\ell'(x)} \left(\frac{f(x)}{u_0(x)} - \lambda\right) \quad \textup{for all } x \neq x_\ell.
	\end{equation}
	The properties of $\varphi_\ell$ allow us to write $\frac{1}{2\varphi_\ell'(x)} = \frac{\widetilde{g}(x)}{x - x_\ell}$ for some $\widetilde{g} \in C^\infty(\R)$ with $\widetilde{g}(x_\ell) \neq 0$. Moreover, using Lemma \ref{lem:WKB_ordre0}, we can write $\frac{f(x)}{2\varphi'_\ell(x) u_0(x)} = \frac{\widetilde{f}(x)}{(x - x_\ell)^{n}}$ for some $\widetilde{f} \in C^\infty(\R)$. Under this notation, \eqref{eq:kind_of_residue} becomes 
	\begin{equation}
	\label{eq:dxv}
	\partial_x v = \frac{\widetilde{f}(x)}{(x - x_\ell)^n} - \lambda \frac{\widetilde{g}(x)}{x-x_\ell}.
	\end{equation}
	Writing 
	$$\widetilde{f}(x) = \sum_{k = 0}^{n-1} \frac{\widetilde{f}^{(k)}(x_\ell)}{k!}(x-x_\ell)^k + (x - x_\ell)^{n} r_{1}\,, \quad \widetilde{g}(x) = \widetilde{g}(x_\ell) + (x - x_\ell) r_{2}(x)$$
	with $r_1,r_2 \in C^\infty(\R)$, and choosing 
	$$\lambda := \frac{\widetilde{f}^{(n-1)}(x_\ell)}{(n-1)!\,\widetilde{g}(x_\ell)}$$ 
	to cancel the ``residue'' term in $(x - x_\ell)^{-1}$ in $\partial_x v$ (corresponding to $k = n-1$) we can then exhibit a solution $v$ of \eqref{eq:dxv} given by
	$$v = \frac{1}{(x-x_\ell)^{n-1}}\sum_{k = 0}^{n-2} \frac{\widetilde{f}^{(k)}(x_\ell)}{k!(k-n+1)}(x - x_\ell)^{k} + R_1 + R_2.$$
	where $R_1,R_2 \in C^\infty(\R)$ satisfy $R_1'= r_1$ and $R_2' = r_2$. In particular, $u_0 v \in C^\infty(\R)$ since $\frac{u_0}{(x - x_\ell)^{n-1}} \in C^\infty(\R)$ by Lemma \ref{lem:WKB_ordre0}. This concludes the proof. 
	\end{proof}
	\begin{lemma}[{Lower bound on the WKB quasi-modes}]
		\label{lem:normWKB}
		Let $m \in \mathbb{N}$ and let $p(x,h)$ be a polynomial in both variables $x$ and $h$ such that 
		$$\frac{\partial^m p}{\partial x^m}(0,0) \neq 0.$$
		Let $u(x,h)$ satisfy 
		$$u(x,h) = p(x-x_\ell,h) + (x - x_\ell)^N v(x,h)$$
		for some $N > m$, where $v$ is bounded on $\R \times [0,1]$.   
		Then, for any compact interval $I \subset \R$ containing $x_\ell$ in its interior, there exists $c > 0$ and $h_0 > 0$ such that 
		$$ \|e^{-\varphi_\ell/h}u\|_{L^2(I)} \geq  c h^{m/2+\frac14} \quad h \in (0,h_0).$$
		\end{lemma}
		We follow the proof of \cite[Lemma B.3]{DR25}).
		\begin{proof} 
			Writing
			$$\|e^{-\varphi_\ell/h}u\|_{L^2(I)} \geq \underbrace{\|e^{-\varphi_\ell/h}p(\cdot - x_\ell,h)\|_{L^2(I)}}_{I_1(h)} - \underbrace{\|e^{-\varphi_\ell/h} (x-x_\ell)^N v\|_{L^2(I)}}_{I_2(h)},$$
			we show that $I_1(h) \geq c h^{\frac{m}{2} + \frac14}$ and $I_2(h) \leq Ch^{\frac{N}{2} + \frac14}$. 
			
			Choosing $\delta > 0$ such that $(x_\ell-\delta,x_\ell+ \delta) \subset I$, we have 
			\begin{align*}
			I_1(h)^2 \geq \int_{|x-x_\ell|\leq \delta} e^{-2\Re(\varphi_\ell(x))/h}|u(x;h)|^2\,dx.
			\end{align*}
			Moreover, since $\textup{Re}(\varphi_\ell(x)) \leq c(x-x_\ell)^2$ for some $c > 0$, using the change of variables $y = \frac{x-x_\ell}{\sqrt{h}}$ leads to
			\begin{align*}
			I_1(h)^2 &\geq  h^{1/2} \int_{-\delta/\sqrt{h}}^{\delta/\sqrt{h}} e^{-cy^2}  \left|p(\sqrt{h} y,h)\right|^2dy\\
			& = h^{1/2} \int_{-\delta/\sqrt{h}}^{\delta/\sqrt{h}} e^{-cy^2}  \left|q(y,\sqrt{h})\right|^2dy
			\end{align*}
			Here, $q(y,\eta) := p(\eta y,\eta^2)$ is a polynomial in $y$ and $\eta$ that we may rewrite under the form
			$$q(y,\eta)= \sum_{r = r_0}^{R} \eta^{r} Q_r(y)\,,  \quad Q_{r_0} \neq 0$$
			for some $r_0 \geq 0$ and some polynomials $Q_r$. Then, 
			$$|q(y,\eta)|^2 = \eta^{2r_0} |Q_{r_0}(y)|^2 + \sum_{r = r_0+1}^{2R} \eta^r \widetilde{Q}_r(y)$$
			for some other polynomials $\widetilde{Q}_r$. The dominated convergence theorem then gives
			$$I_1(h)^2 = C_{r_0}h^{r_0+\frac12} + \mathscr{O}(h^{r_0 + 1})$$
			where $C_{r_0} = \int_{\R} e^{-(1\pm\varepsilon)cy^2} |Q_{r_0}(y)|^2\,dy \neq 0$ (since $Q_{r_0} \neq 0$). Finally, we observe that 
			$$\frac{\partial^{2m}q}{\partial y^m \partial \eta^m}(y,\eta) = \frac{\partial^m}{\partial \eta^m} \left(\frac{\partial^m}{\partial y^m} [p(\eta y,\eta^2)]\right) = \frac{\partial^m}{\partial \eta^m} \left(\eta^m \frac{\partial^m p}{\partial x^m} (\eta y,\eta^2)\right)$$
			and thus 
			$$\frac{\partial^m Q_m}{\partial y^m}(0) = \frac{\partial^{2m}q}{\partial \eta^m \partial y^m}(0,0) =  \frac{\partial^{2m}q}{\partial y^m \partial \eta^m}(0,0) = \frac{\partial^m p}{\partial x^m}(0,0) \neq 0.$$
			In particular, $Q_m \neq 0$, which implies that $r_0 \geq m$, and thus, $I_1(h) \geq c h^{m+ \frac14}$ for $h$ small enough. 
			
			To estimate $I_2(h)$, we first choose $\delta > 0$ small enough such that $\varphi_\ell(x) \geq c(x - x_\ell)^2$ for all $|x-x_\ell|\leq \delta$, with $c > 0$. Since $\varphi_\ell(x) \geq c' > 0$ for $|x-x_\ell|\geq \delta$, we thus have for $h \leq 1$, 
			\begin{align*}
			I_2(h)^2& \leq C\int_{|x-x_\ell| \leq \delta} e^{-c\frac{(x-x_\ell)^2}{h}} (x - x_\ell)^{2N}\,dx + C e^{-2c'/h} \int_{I} (x-x_\ell)^{2N}\,dx\\
			& \leq h^{N+\frac12}\int_{\R} e^{-cy^2} y^{2N}\,dy + \mathscr{O}(h^\infty),
			\end{align*}
			which concludes the proof. 
		\end{proof}
		\begin{lemma}[WKB normalization constants for $n = 1$]
		\label{lem:constants_n=1}
		Let $\psi_1^{\rm wkb}$ be a WKB quasimode with parameter $(1,J)$, and let $I \subset \R$ be open an containing $x_\ell$. Then 
		\begin{equation}
		\label{eq:normalization_1}
		\|\psi_1^{\rm wkb}(\cdot;h)\|_{L^2(I)} \sim \left(\frac{\pi}{\cos(\alpha/2)}\sqrt{\frac{2}{V''(x_\ell)}}\right)^{1/4}h^{1/4},
		\end{equation}
		\begin{equation}
		\label{eq:normalization_2}
		\langle \psi_1^{\rm wkb},\overline{\psi_1^{\rm wkb}}\rangle  \sim e^{-i\alpha/4} \left(\pi \sqrt{\frac{2}{V''(x_\ell)}}\right)^{1/2} h^{1/2}
		\end{equation}
		\end{lemma}
		\begin{proof}
		One can check using the same techniques as in the proof of Lemma \ref{lem:normWKB} that 
		$$\|\psi_1^{\rm wkb}\|^2_{L^2(I)} =  \int_{I} e^{-2\Re(\varphi_\ell(x))/h} |a(x;h)|^2\,dx \sim h^{1/2}|a_{1,0}(x_\ell)|^2 \int_{\R} e^{- \Re(\varphi_\ell''(x_\ell))y^2}\,dy\,.$$
		Since $a_{1,0}(x_\ell)=1$, we get 
		$$\|\psi_1^{\rm wkb}\|^2_{L^2(I)} \sim h^{1/2} \sqrt{\frac{\pi}{\Re(\varphi_\ell''(x_\ell))}}$$
		and \eqref{eq:normalization_1} follows since $\Re(\varphi_\ell''(x_\ell)) = \cos(\alpha/2) \sqrt{V''(x_\ell)/2}.$
		
		Similarly,
		\begin{align*}
		\langle \psi_1^{\rm wkb},\overline{\psi_1^{\rm wkb}}\rangle &= \int_{\R} e^{-2\varphi_\ell(x)/h} a(x,h)^2\,dx \sim h^{1/2} \int_{\R} e^{-e^{i\alpha/2} \sqrt{\frac{V''(x_\ell)}{2}}y^2}dy \\
		&= h^{1/2} e^{-i\alpha/4}\sqrt{\pi} \left(\frac{2}{V''(x_\ell)}\right)^{1/4},
		\end{align*}
		where we used the fact that $\int_{\R} e^{-zx^2}\,dx = \sqrt{\frac{\pi}{z}}$ for $\textup{Re}(z) > 0$, with $\sqrt{\cdot}$ denoting the principal determination of the square root. This shows \eqref{eq:normalization_2}. 
		\end{proof}

	\section{An estimate for a Riesz projector.}

	\begin{proposition} 
	\label{prop:general_spectral_result}
	Let $H$ be a Hilbert space and $A : H \to H$ a densely defined closed operator on $H$. Let $\Omega \subset \C$ be a simply connected, bounded open set whose boundary is described by a regular closed simple curve $\gamma:[0,1] \to \mathbb{C}$ and let $K \subset \Omega$ be compact. Suppose that $\partial \Omega$ does not intersect the spectrum of $A$, and let 
	$$P := \frac{1}{2\pi i}\int_{\gamma} (z - A)^{-1}\,dz.$$ 
	Then for all $u \in \textup{dom}(A)$ and all $z_0 \in K$, 
	$$\|(\Id - P)u\| \leq \frac{|\gamma|}{2\pi \textup{dist}(K,\partial \Omega)}\left( \sup_{z \in \partial \Omega} \|(A-z)^{-1}\| \right) \|(A - z_0)u\|\,.$$
	\end{proposition}
	\begin{proof}
	By the Cauchy formula, 
	\begin{align*}
	(\Id - P) u &= \rev{\frac{1}{2\pi i}}\int_{\gamma} \left((z-z_0)^{-1} - (z - A)^{-1}\right)\,dz\\
	& = \rev{\frac{1}{2\pi i}}\int_\gamma (z-z_0)^{-1}(A - z_0)(z-A)^{-1}\\
	& =: S(z_0) (A - z_0)
	\end{align*}
	where 
	$$S(z_0) := \frac{1}{2\pi i} \int_\gamma (z - z_0)^{-1} (z - A)^{-1}\,dz.$$
	The result is obtained by a direct estimate of the integral defining $S$.  
	\end{proof}

		\bibliographystyle{abbrv}
		\bibliography{biblio.bib}

\end{document}